\documentclass[a4paper,12pt,reqno]{amsart}

\usepackage{a4}
\usepackage[utf8]{inputenc}
\usepackage[english]{babel}
\usepackage{verbatim}
\usepackage{comment}
\usepackage{color}

\usepackage{amssymb,stmaryrd}
\usepackage{color}
\usepackage{graphicx}
\usepackage{floatflt}
\usepackage{float} 		

\usepackage{amssymb}
\usepackage{amsthm}
\usepackage{amscd}
\usepackage{amsmath} 
\usepackage{mathtools}
\usepackage[colorlinks=true]{hyperref}
\usepackage[all]{xy}
\usepackage{url}

\newtheorem{definition}{Definition}
\newtheorem{properties}[definition]{Properties}
\newtheorem{example}[definition]{Example}

\newtheorem{theorem}[definition]{Theorem}
\newtheorem*{theorem*}{Theorem}
\newtheorem{lemma}[definition]{Lemma}
\newtheorem*{lemma*}{Lemma}
\newtheorem{remark}[definition]{Remark}

\newtheorem{corollary}[definition]{Corollary}
\newtheorem*{corollary*}{Corollary}

\newtheorem{proposition}[definition]{Proposition}

\def\XXint#1#2#3{{\setbox0=\hbox{$#1{#2#3}{\int}$}
		\vcenter{\hbox{$#2#3$}}\kern-.5\wd0}}

\makeatletter
\@addtoreset{definition}{section}
\@addtoreset{equation}{section}
\makeatother

\newcommand{\cS}{\mathcal{S}}

\DeclareMathOperator{\Res}{\mathrm{Res}}

\title{Quantum Curves in the Context of Symplectic Duality}
\author{Alexander Hock}

\address{Institute for Mathematics, University of Heidelberg, Mathematikon, Berliner Str. 41-49,
	69120 Heidelberg, Germany} 
\email{alexander.hock@uni-heidelberg.de}

\author{Sergey Shadrin}

\address{Korteweg-de Vries Institute for Mathematics, University of Amsterdam, P.O. Box 94248, 1090GE Amsterdam, The Netherlands}

\email{s.shadrin@uva.nl}

\makeatletter
\def\@tocline#1#2#3#4#5#6#7{\relax
  \ifnum #1>\c@tocdepth 
  \else
    \par \addpenalty\@secpenalty\addvspace{#2}%
    \begingroup \hyphenpenalty\@M
    \@ifempty{#4}{%
      \@tempdima\csname r@tocindent\number#1\endcsname\relax
    }{%
      \@tempdima#4\relax
    }%
    \parindent\z@ \leftskip#3\relax \advance\leftskip\@tempdima\relax
    \rightskip\@pnumwidth plus4em \parfillskip-\@pnumwidth
    #5\leavevmode\hskip-\@tempdima
      \ifcase #1
       \or\or \hskip 1em \or \hskip 2em \else \hskip 3em \fi%
      #6\nobreak\relax
    \hfill\hbox to\@pnumwidth{\@tocpagenum{#7}}\par
    \nobreak
    \endgroup
  \fi}
\makeatother

\begin{document}

\begin{abstract} We discuss how to use the recent progress in understanding of the $x$-$y$ duality and symplectic duality in the theory of topological recursion and its generalizations in order to efficiently compute the quantum spectral curve operators for the wave functions with arbitrary base points. The paper also contains an overview of recent generalizations of the setup of topological recursion prompted by the progress in understanding the $x$-$y$ duality
\end{abstract}

\maketitle

\tableofcontents

\section{Introduction}

Topological recursion (TR) was proposed in its original form in the works of Chekhov, Eynard, and Orantin~\cite{Chekhov:2006rq,Eynard:2007kz} as a computational tool for the cumulants of the matrix models. In the recent years it has appeared to be a universal and ubiquitous interface between enumerative geometry, combinatorics, and integrability, with further applications in knot theory, free probability theory, mirror symmetry, representation theory, and many other areas of mathematics and mathematical physics. It is known to substantially enrich and facilitate strong interaction of all areas where it is applied and was instrumental in settling long standing open questions. 

The input data of topological recursion consists of a complex curve $\Sigma$, which is called the spectral curve, with some additional structure. Symplectic duality represents a system of intrinsic transformations of the input data of topological recursion that preserves the spectral curve. These transformations appear to be very powerful computational tools themselves, but also they prompt various extensions of the original setting of topological recursion. To this end, there is a variety of generalizations of the original setting of topological recursion, motivated both by compatibility with the symplectic duality and by further applications.

Topological recursion produces a system of symmetric differentials on the Cartesian powers of the underlying spectral curve. There is a system of further natural (from the point of view of symplectic duality) objects called extended (half-) differentials associated to a system of symmetric differentials, which include the so-called  kernel and its specialization called the wave function. The latter ones have independent interest in applications. 
A remarkable and quite general result states that there is a differential equation for the wave function that one can obtain as a quantization of the underlying spectral curve.

The key observation that we promote in this paper is that the theory of symplectic duality provides concrete tools to derive the quantum spectral curves in a variety of important examples. This approach was already discovered by Weller in~\cite{Weller:2024msm} in application to wave functions under some extra restrictive assumptions, but we extend this type of arguments to the whole kernel. Along this way we treat new examples in the generalized setting as well as conceptually revisit and substantially simplify the derivation in a number of important examples already covered in the literature.  

\subsection{Organization of the paper}

Section~\ref{sec:Bacxground} contains an overview of various species of topological recursion and the context where these extensions emerge in a natural way, namely, the theory of $x$-$y$ duality and more general symplectic duality. We also recall the concept of quantum spectral curve and basic facts about it. 

In Section~\ref{sec:QuantumCurve-dualities} we present general formulas of the action of the $x$-$y$ and symplectic dualities on the quantum spectral curve operators (and, in some cases, their dual cousins), which prompts an efficient derivation algorithm.

Section~\ref{sec:examples} contains an extensive list of examples. Many of them concern re-derivation of the quantum curves already derived by different methods in the literature, and yet we got with our approach much simpler formulas with an arguably more conceptual derivation.

\subsection{Acknowledgments} A.~H. was supported through the Walter-Benjamin fellowship\footnote{\
	``Funded by
	the Deutsche Forschungsgemeinschaft (DFG, German Research
	Foundation) -- Project-ID 465029630 and Project-ID 51922044''} and
through the project \ ``Topological Recursion, Duality and
Applications''\footnote{\ ``Funded by
	the Deutsche Forschungsgemeinschaft (DFG, German Research
	Foundation) -- Project-ID 551478549''}. 
S.~S. was supported by the Dutch Research Council.

\section{Background}

\label{sec:Bacxground}
This section provides background on some recent developments as well as well-known facts about topological recursion. First, we will review the original formulation of TR. After that, we will discuss various extensions of TR that generalize it in different directions. Some of these generalizations are very recent and still under exploration. We then review the application of the original formulation of TR to the quantum curve and the construction of a wave function. Notably, the construction of the quantum curve is not yet established for the different types of TR generalizations. In the last subsection of this section, we will review the recent development of the $x$-$y$ duality, which is already related to TR generalizations. In fact, the $x$-$y$ duality motivated the definition of some of these generalizations.

\subsection{CEO-TR}
The original definition of TR arose from considerations in matrix models \cite{Chekhov:2006rq}, but it was clear from the beginning that the definition could be formulated in a more general setting \cite{Eynard:2007kz}. Due to its origin, we will call this version CEO-TR after Chekhov, Eynard and Orantin. 

The basic idea is that TR can be viewed as a (universal) procedure that takes some initial data as input and generates an infinite family of multi-differential forms, $\omega_{g,n}$, as output. Depending on the input, the output can have various fascinating applications and interpretations. For completeness, we will provide the original formulation of TR, though it will not be used explicitly in this article.

The input is the spectral curve data, a tuple $(\Sigma, x, y, B)$, where $\Sigma$ is a \textup{compact} Riemann surface, and $x, y\colon \Sigma \to \mathbb{C}$ are \textup{meromorphic functions} such that $x$ has \textup{simple} ramification points on $\Sigma$, and $y$ is regular and $dy$ is non-vanishing at the ramification points of $x$. The bi-differential $B$ is symmetric, holomorphic outside the diagonal, with a double pole on the diagonal 
with bi-residue $1$. For a choice of $\mathfrak{A}$- and $\mathfrak{B}$-cycles on $\Sigma$, the bi-differential $B$ is uniquely determined by normalizing along the $\mathfrak{A}$-cycles.

\begin{definition}[\cite{Eynard:2007kz}] \label{def:claTR} 
    CEO-TR generates a family of multi-differentials $\omega_{g,n}$ on $\Sigma^n$, $g\geq 0$, $n\geq 1$, with $\omega_{0,1}=y\,dx$ and $\omega_{0,2}=B$. Furthermore, for negative Euler characteristic $\chi=2-2g-n<0$, all $\omega_{g,n}$ are defined recursively in local coordinates via 
\begin{align}
  \label{eq:TR-intro}
&  \omega_{g,n+1}(I,z)
  =\sum_{p_i\in Ram(x)}
  \Res\displaylimits_{q\to p_i} \frac{\frac{1}{2}\int^{q}_{\sigma_i(q)}
    B(z,\bullet)}{\omega_{0,1}(q)-\omega_{0,1}(\sigma_i(q))}
  \bigg(
  \omega_{g-1,n+2}(I, q,\sigma_i(q))\\
  &\qquad \qquad\qquad\qquad\qquad\qquad
  +
   \sum_{\substack{g_1+g_2=g\\ I_1\uplus I_2=I\\
            (g_i,I_i)\neq (0,\emptyset)}}
   \omega_{g_1,|I_1|+1}(I_1,q)
  \omega_{g_2,|I_2|+1}(I_2,\sigma_i(q))\!\bigg). \nonumber
\end{align}
The following notation is used:
\begin{itemize}
\item  $I=\{z_1,\dots,z_n\}$ is a collection of $n$ local coordinates $z_j$, and $I_1,I_2\subset I$ are any disjoint subsets of $I$ with union to be $I=I_1\uplus I_2$
\item $Ram(x)\coloneqq \{p_1,\dots,p_N\}$, $N\geq 0$, is the set of ramification points of $x$ 
\item  the local Galois involution $\sigma_i\neq \mathrm{id}$ with
$x(q)=x(\sigma_i(q))$ is defined in the vicinity of $p_i\in Ram(x)$ with fixed point $p_i=\sigma_i(p_i)$
\end{itemize}
\end{definition}
\begin{remark}\label{Rem:ClasTR}
    We highlight a few properties specific to the formula \eqref{eq:TR-intro} and to the notion of \textup{CEO-TR}. The Riemann surface $\Sigma$ is assumed to be compact and connected, both functions $x$ and $y$ are assumed to be meromorphic, ramification points of $x$ are assumed to be simple, and $y$ is regular and $dy$ is non-vanishing at the ramification points of $x$.

However, in most applications of CEO-TR, some of the properties above are not satisfied, but the recursion still works and produces meaningful answers. 
\end{remark}
\begin{remark}
    The formula \eqref{eq:TR-intro} can, of course, also be used for non-compact and/or disconnected Riemann surfaces, consisting, for instance, of open discs defined around the simple ramification points of $x$. However, in this situation, $\omega_{g,n}$ is only locally defined, with no global extension.

    Another option is to have $x$ and $y$ just locally defined on some union of open disks around points in $Ram(x)\subset \Sigma$, and $B$ being defined on the whole Riemann surface $\Sigma$, still not necessarily compact. In this case $\omega_{0,1}$ is defined locally, but all other $\omega_{g,n}$'s are constructed globally. 
\end{remark}
We want to collect some properties that follow more or less directly from the definition, see \cite{Eynard:2007kz} for the details.
\begin{properties}\label{properties:omega}
For $2g+n-2> 0$, the $\omega_{g,n}$'s
    \begin{itemize}
    \item are symmetric, i.e. $\omega_{g,n}(...,z_i,..z_j,..)=\omega_{g,n}(...,z_j,..z_i,..)$ for all $i,j$;
    \item have poles only located at the ramification points $p\in Ram(x)$;
    \item have vanishing residue at the ramification points $p\in Ram(x)$;
    \item are homogeneous, i.e. if we multiply $\omega_{0,1}$ by some non-vanishing constant $\lambda$ that is  $\omega_{0,1}\mapsto \lambda \omega_{0,1} $ then $\omega_{g,n}\mapsto \lambda^{2-2g-n}\omega_{g,n}$
    \item are normalized along the $\mathfrak{A}$-cycles, i.e. $\oint_{z_i\in \mathfrak{A}_j}\omega_{g,n}(...,z_i,..)=0$ for all $i=1,\dots,n$, $j=1,\dots,g$;
    \item are invariant under the following symplectic transformation of $x$ and $y$
    \begin{align*}
        &(x,y)\quad \mapsto\quad  (x,y+R(x))\quad \text{for some rational function $R$}\\
        &(x,y)\quad \mapsto\quad  \big(\frac{ax+b}{cx+d},(cx+d)^2y\big)\quad \text{ with }\begin{pmatrix}
		a & b \\
		c & d 
		\end{pmatrix}\in SL_2(\mathbb{C}).
    \end{align*}
    (N.B.: in both cases there are some extra open conditions on the parameters of transformations.)
\end{itemize}
\end{properties}
The name of symplectic transformation originates from the fact that the formal symplectic form $dx\wedge dy$ is invariant under the mentioned transformations. 

\begin{remark} Under some extra assumptions on $x$ and $y$, namely, if all ramification points of $y$ are simple and $x$ is regular and $dx$ is non-vanishing at $dy=0$, we can consider a third symplectic transformation not leaving $\omega_{g,n}$ invariant: $(x,y)\to (-y,x)$.  
\end{remark}

Even under the most restrictive assumptions CEO-TR provides already very important mathematics as an output in terms of $\omega_{g,n}$.
\begin{example}
	\label{ex:KontWitt}
    Let the spectral curve be 
    $$(\Sigma =\mathbb{P}^1, x = {z^2}, y = z, B=\frac{dz_1 , dz_2}{(z_1 - z_2)^2}).$$ 
    Then the output of CEO-TR encodes the intersection numbers of the $\psi$-classs on the moduli spaces of complex curves $\overline{\mathcal{M}}_{g,n}$. 
\end{example}

\subsection{Extensions of topological recursion}\label{Sec:ExtTR}

It is very natural to ask how TR behaves locally under certain deformations and limits \cite{Borot:2023wik} in which, for instance, ramification points collide and give rise to higher-order ramification. Clearly, Definition \ref{def:claTR} cannot be used  verbatim if we have a ramification point of higher degree.

The meromorphicity of $x$ and $y$ is also violated in crucial applications. Requiring only that $dx$ and $dy$ be meromorphic is more general and may allow for logarithmic behavior in $x$ and $y$. For instance, supported by explicit computations in perturbation theory, an application of CEO-TR was conjectured and then proved for the $B$-model topological string theory of toric Calabi-Yau 3-folds (the BKMP conjecture, \cite{Bouchard:2007ys}, the proofs are available in \cite{Chen:2009ws,Zhou:2009ea,Eynard:2012nj,Fang:2016svw}). This conjecture uses formula \eqref{eq:TR-intro} but with a curve of the form $Pol(e^x, e^y) = 0$. In other words, $x,y$ assumed not to be meromorphic. However, the conjecture should hold for \emph{generic framing}, which in our language means after a suitable symplectic transformation. 

Consequently, CEO-TR is just the tip of an iceberg, where several different (but rigid) types of extensions need to be explored. 
Four different type of extensions are reviewed in this subsection: 
\begin{itemize}
	\item irregular topological recursion of Chekhov-Norbury (CN-TR);
	\item topological recursion for higher order ramifications of Bouchard-Eynard (BE-TR);
	\item logarithmic topological recursion (Log-TR);
	\item generalized topological recursion (Gen-TR).
\end{itemize}
We will not give explicit formulas, but rather explain the properties those extensions have, and refer to the literature for precise definitions. The interrelation of these extensions is pictorially presented in Figure~\ref{Gen-TR}:

\begin{figure}[h]
	\includegraphics[width=0.6\textwidth]{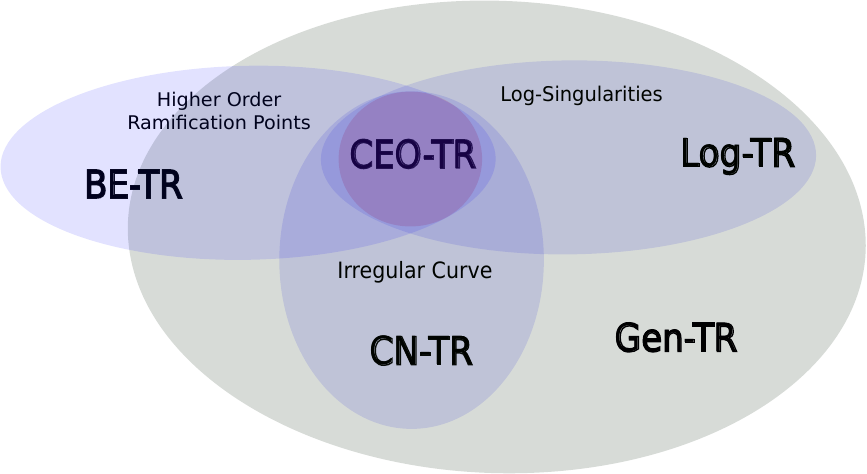}
	\caption{A schematic overview of the different versions of topological recursion. The most recent extension, Gen-TR, reproduces CEO-TR, Log-TR and CN-TR in general and BE-TR under some additional assumption.}
	\label{Gen-TR}
\end{figure}

\subsubsection{Irregular topological recursion}
The term \textit{irregular TR} refers to the behavior of $y$ around the ramification points $Ram(x)$. We relax the condition that $y$ is regular at the ramification points of $x$, now allowing $y$ to have a simple pole at $p_i \in Ram(x)$.  Consequently, $\omega_{0,1} = y \, dx$ remains regular at these ramification points. We still assume if $y$ is regular at $p_i\in Ram(x)$ to have $dy(p_i)\neq 0$.

This type of spectral curves has been studied in various examples, such as in \cite{Do:2014ncn,Do:2016odu}, with a most general exposition in~\cite{Chekhov:2017aot}. To be consistent with nomenclature, we refer to irregular spectral curves with simple ramification points of $x$ as \textit{CN-TR} after Chekhov and Norbury. The spectral curve consists of the same initial data $(\Sigma, x, y, B)$ as in CEO-TR, and the recursive definition of the multi-differentials coincides with \eqref{eq:TR-intro}. All Properties \ref{properties:omega} concerning the multi-differentials $\omega_{g,n}$ remain satisfied in the case of CN-TR.

Importantly, irregular spectral curves led to the definition of new cohomology classes on the moduli space of complex curves~ \cite{Norbury:2017eih}:
\begin{example}\label{ex:Norbury}
	Let the spectral curve be 
	$$(\Sigma =\mathbb{P}^1, x = {z^2}, y = \frac 1z, B=\frac{dz_1 , dz_2}{(z_1 - z_2)^2}).$$ 
	Then the output of CN-TR encodes the intersection numbers of the $\psi$-classes with the so-called Norbury $\Theta$-class on $\overline{\mathcal{M}}_{g,n}$. 
\end{example}

\subsubsection{Higher order topological recursion}
The aim is to construct again a local formula which computes 
a family of multi-differentials $\omega_{g,n}$ where the ramified covering $x$ can have higher order ramifications. The correct definition should be compatible with taking limits of CEO-TR of Definition \ref{def:claTR} with colliding ramification points. An extension in this sense was suggested by Bouchard and Eynard \cite{Bouchard:2012yg} and is called \textit{BE-TR}. 

The explicit definition of BE-TR will not be recalled here, since the definition is quite involved. We refer to \cite{Bouchard:2012yg} for explicit formulas. Nevertheless, we want to emphasize important facts and differences. The initial data of BE-TR consists again of the spectral curve $(\Sigma,x,y,B)$ now allowing $x$ to have higher order ramification points $Ram(x)\coloneqq \{p_1,\dots,p_N\}$, where $y$ is assumed to be regular and $dy$ non-vanishing at $p_i\in Ram(x)$. The defining formula for $\omega_{g,n}$ in BE-TR depends on the order of the ramification points $p_i\in Ram(x)$. The formula is not necessarily quadratic as in \eqref{eq:TR-intro}, nor is the kernel the same, nor is the integrand of Euler characteristic $\chi + 1$ if we want to compute $\omega_{g,n}$ with Euler characteristic $\chi = 2 - 2g - n$.
These drastic changes of the formula in BE-TR incorporate the fact that  the local Galois action at 
a ramification point $p_i \in \text{Ram}(x)$ of higher order is no longer generated by an involution.

BE-TR provides as an output multi-differentials $\omega_{g,n}$ satisfying again Properties \ref{properties:omega}. Moreover, BE-TR is compatible with CEO-TR in the limit of colliding ramification points, see \cite{Bouchard:2012yg,Borot:2023wik}.  
The prime example of BE-TR reads: 
\begin{example}
     Let the spectral curve be $$(\Sigma=\mathbb{P}^1,x=z^{r},y=z,B=\frac{dz_1\,dz_2}{(z_1-z_2)^2}).$$ 
     Then the output of BE-TR encodes the intersection numbers of the $\psi$-classes with the so-called $r$-spin Witten class on $\overline{\mathcal{M}}_{g,n}$.
\end{example}

\begin{remark} \label{rem:BE-irreg}
    For BE-TR, we can also relax the condition on $y$ at $p_i \in \text{Ram}(x)$, allowing $y$ to have a simple pole or to remain regular with $dy(p_i) \neq 0$. However, generalizing the behavior of $y$ at $p_i \in \text{Ram}(x)$ in BE-TR by allowing $y$ to have higher-order poles 
    poses serious issues: the formula of BE-TR with higher order poles 
    of $y$ at $p_i \in \text{Ram}(x)$ can generate multi-differentials $\omega_{g,n}$ that are {not} symmetric~\cite{MR4744795}, which is one of the most desired properties in \ref{properties:omega}. The latter problem can be fixed by using the generalized TR that we discuss below, see Sections~\ref{sec:Gen-TR} and~\ref{sec:rs-curve}.
\end{remark}

\subsubsection{Logarithmic topological recursion}\label{sec:logTR}
The extension of TR taking care properly of logarithmic poles is called \textit{Log-TR}, and was defined in~\cite{Alexandrov:2023tgl-LogTR}, inspired by earlier insights related to the $x$-$y$ duality observed in~\cite{Hock:2023dno}.  
The initial data for Log-TR is again the spectral curve of the form $(\Sigma,x,y,B)$, now allowing more general $x,y$. This extension of TR will include local logarithms for $x$ and $y$ such that $dx$ and $dy$ are meromorphic; we are still assuming that the ramification points of $x$ are again simple. From another perspective we can say that $dx$ and $dy$ are now allowed to have non-vanishing residues. 

\begin{remark}
In principle, one might just take the original CEO-TR and allow $x$ and $y$ to have logarithmic singularities, but this is incompatible with the limiting procedures and symplectic transformations. 

An important fact is that Log-TR is equal to CEO-TR if $y$ has no logarithmic poles or if the logarithmic poles of $y$ are canceled out by poles of $x$.
However, if no cancellation happens CEO-TR would provide misleading results (see e.g.~ \cite{Bouchard:2011ya}).	
\end{remark}

The difference between CEO-TR and Log-TR appears when so-called Log-TR-vital logarithmic singularities exist. We say that the primitive $y$ of the differential $dy$ has \textit{logarithmic
singularity} at some point $a$ on $\Sigma$ if $dy$ has a pole at $a$ with nonzero residue. A logarithmic singularity of $y$ is called \textit{Log-TR-vital} if this pole of $dy$ is simple and $dx$ has no pole at this point.

Let $a_1, \dots , a_M$ be the Log-TR-vital singular points of $y$. We denote the residues of $dy$
at these points by $\alpha_1^{-1},\dots,\alpha_M^{-1}$, respectively. That is, the principal part of $dy$ near $a_i$ is
given by $\alpha_i^{-1}dz/(z-a_i)$ in any local coordinate $z$.
Consider the germs of the principal parts of meromorphic 1-forms in the neighborhoods
of $a_i, i = 1, . . . , M$, defined as the principal parts of coefficients of positive powers of $\hbar$ in
the following expressions:
\begin{align}
    \left(\frac{1}{\alpha_i S(\alpha_i \hbar \partial_x)}\log (z-a_i)\right)dx,
\end{align}
where $S(t)=\frac{e^{t/2}-e^{-t/2}}{t}=\sum_{k=0}^\infty \frac{t^{2k}}{4^k(2k+1)!}$ understood as formal power series and $\partial_x=\frac{d}{dx}=\frac{1}{x'(z)}\frac{d}{dz}$ in some local coordinate $z$.
Then, Log-TR is defined by adding an extra summand
\begin{align}
	\label{eq:TR-introLog}
	\delta_{n,0}\sum_{i=1}^M\Res\displaylimits_{q\to a_i}\int_{a_i}^q \omega_{0,2}(z, \bullet)[\hbar^{2g}]\left(\frac{1}{\alpha_i S(\alpha_i \hbar \partial_{x(q)})}\log (q-a_i)\right)dx(q).
\end{align} 	
one the right hand side of the recursion~\eqref{eq:TR-intro}.

The output $\omega_{g,n}$ of Log-TR satisfies all but two properties \ref{properties:omega}. First, $\omega_{g,n}$ for $(g,n)=(g,1)$ have additionally poles at the Log-TR-vital points $a_i$. Second, the invariance under the $SL_2$-symplectic transformation $(x,y)\mapsto  \big(\frac{ax+b}{cx+d},(cx+d)^2y\big)$ can generate a $dy$ that is not meromorphic anymore making Log-TR inapplicable. 

\begin{remark}
    If $x$ is unramified, Log-TR produces for $2g+n-2>0$ vanishing $\omega_{g,n}$ for $n\geq 2$. However, all $\omega_{g,1}$ are non-trivial as long as Log-TR-vital singular points exist. 
\end{remark}

Some examples for Log-TR can be found in \cite{Hock:2023dno,Alexandrov:2023tgl-LogTR,Alexandrov:2024ajj}. In the context of the BKMP conjecture mentioned above, the distinction between Log-TR-vital and non-vital singularities is related to the choice of framing. To this end, one considers a basis spectral curve $H(e^x,e^y)=0$, $H$ is a polynomial, and its framed version  $H_f(e^x,e^y)=0$ obtained by a symplectic transformation $(x,y)\mapsto (x+fy,y)$, $f\in\mathbb{C}$. For a general choice of framing $f$, the logarithmic singularities of $y$ are non-vital, hence CEO-TR can be applied instead of Log-TR. However, it was observed in \cite{Bouchard:2011ya} that CEO-TR does not reprodece the B-model for some explicitly chosen framings, and these are exactly the cases where the logarithmic singularities of $y$ are Log-TR-vital, hence one have to apply Log-TR.

\begin{remark}
    Log-TR is a generalization of CEO-TR, but it has been argued \cite{Alexandrov:2023tgl-LogTR,Alexandrov:2024ajj} that Log-TR is \textup{the canonical extension} of CEO-TR which is necessary to include logarithmic singularities in a proper way. This article further supports this statement by the previous discussion and additionally in terms of looking at the corresponding quantum curve such that Log-TR generates the expected wave function even if $x$ and/or $y$ have no ramification points.
\end{remark}

\subsubsection{Generalized topological recursion}
\label{sec:Gen-TR}
Generalized topological recursion (Gen-TR) was defined in~\cite{Alexandrov:2024tjo}.  It is based on a different idea of input data and the right hand side of a recursion formula generalizes~\eqref{eq:TR-intro} to a universal formula that comes from the $x$-$y$-duality and connections to KP integrability. We explain the input data, referring to~\cite{Alexandrov:2024tjo} for further details. 

For a given point $q \in \Sigma$ and a local coordinate $z$ at this point, consider the local expressions of $dx$ and $dy$: \begin{align*} dx = az^{r-1}(1+\mathcal{O}(z)), dz, \quad dy = bz^{s-1}(1+\mathcal{O}(z)), dz, \quad a, b \neq 0, \quad r, s \in \mathbb{Z}. \end{align*} The point $q$ is called \textit{non-special} if either $r = s = 1$ or $r + s \leq 0$, and \textit{special} otherwise. Ramification points, logarithmic points of $x$ or $y$, and singular points are special unless the condition $r + s \leq 0$ is satisfied. A part of the input of Gen-TR is a choice of  a subset  $\mathcal{P}$ of the set of special points.

    The full input of Gen-TR is a tuple $(\Sigma, dx, dy, B, \mathcal{P})$, where $\Sigma$ is a Riemann surface, $x, y: \Sigma \to \mathbb{C}$ are functions on $\Sigma$ such that $dx$ and $dy$ are meromorphic, $B$ is a symmetric bi-differential with a double pole on the diagonal, bi-residue 1, and no other poles (for a compact $\Sigma$, we normalize it on the $\mathfrak{A}$-cycles). $\mathcal{P}$ is an arbitrary subset of the set of special points.

\begin{definition}[\cite{Alexandrov:2024tjo}] \label{def:GTR}

    Gen-TR generates a family of multi-differentials $\omega_{g,n}$ on $\Sigma^n$, $g\geq 0$, $n\geq 1$, $2g-2+n\geq 0$, with $\omega_{0,2}=B$. Furthermore, for negative Euler characteristic $\chi=2-2g-n<0$, all $\omega_{g,n}$ are defined recursively in local coordinates via 

\begin{align}
  \label{eq:GTR}
  \omega_{g,n+1}(I, z) = \sum_{p \in \mathcal{P}} \Res\displaylimits_{q \to p} \bigg(\int^{q} B(\cdot, z)\bigg) \tilde{\omega}_{g,n+1}(I,q),
\end{align}
where $\tilde{\omega}_{g,n+1}$ is an explicit global expression built recursively from all $\omega_{g',n'}$ with $2g + n - 2 > 2g' + n' - 2$, $dx$, and $dy$, but independent of the order of the local behavior of $x$ and $y$ around the special point $p \in \mathcal{P}$.
\end{definition}

The multi-differentials $\omega_{g,n}$ generated by Gen-TR are symmetric, have poles located only at points in $\mathcal{P}$ (note that $\mathcal{P}$ does not necessarily consist of the ramification points of $x$). Furthermore, the $\omega_{g,n}$ generated by Gen-TR are not necessarily invariant under the two symplectic transformations $(x, y) \mapsto (x, y + R(x))$ and $(x, y) \mapsto \big(\frac{ax + b}{cx + d}, \frac{(cx + d)^2}{ad - bc} y\big)$.

If the spectral curve of Gen-TR satisfies the setup of CEO-TR (respectively, CN-TR, BE-TR, Log-TR) and the set $\mathcal{P}$ is chosen to be the set of zeroes of $dx$ (and, in the case of Log-TR, $\mathcal{P}$ should also include all Log-TR-vital points), then Gen-TR produces the same differentials as CEO-TR (respectively, CN-TR, BE-TR, Log-TR)~\cite{Alexandrov:2024tjo}.

\begin{remark} It is mentioned in Remark~\ref{rem:BE-irreg} that one can consider a version of BE-TR with singularities of $y$ at the critical points of $x$ such that $ydx$ is regular. Under some conditions BE-TR still produces symmetric differentials, and in general they are different from the ones produced by Get-TR even if the set $\mathcal{P}$ is chosen to be the set of the critical points of $x$.
\end{remark}

\begin{example}
    Let the spectral curve of Gen-TR be 
    $$(\Sigma = \mathbb{P}^1, dx = dz^r, dy = dz^{s-r}, B = \frac{dz_1\,dz_2}{(z_1 - z_2)^2}, \mathcal{P} = \{0\}), $$
    where $r=2,3,\dots$ and  $s=1,\dots,r-1$.  	Then the output of Gen-TR encodes the intersection numbers of the $\psi$-classes with the so-called  $\Theta^{(r,s)}$-classes on $\overline{\mathcal{M}}_{g,n}$. The latter classes are the top degree parts of some $\Omega$ classes (also called Chiodo classes), see~\cite{chidambaram2023relationsoverlinemathcalmgnnegativerspin} for $s=r-1$ and ~\cite{CGS} for general $s$. 
    
	Note that in this case the BE-TR (or rather its irregular version as in Remark~\ref{rem:BE-irreg}) can be applied only if $r\equiv \pm 1 \mod s$ (otherwise the differentials it produces are not symmetric) and in most cases the enumerative interpretation of the output it is not known. 
\end{example}


\subsection{Quantum curve from topological recursion}
The quantum curve is a quantization of the original spectral curve in the sense of classical quantum mechanics. For meromorphic functions $x$ and $y$ on a compact Riemann surface $\Sigma$, there exists a polynomial $P(x, y)$ such that 
\begin{align}\label{quantumcurve}
    \{P(x(z), y(z)) = 0 \mid x, y : \Sigma \to \mathbb{C}, \forall z \in \Sigma\}.
\end{align}
A quantization of the complex curve \eqref{quantumcurve} is an $\hbar$-deformed operator $\hat{P}_\hbar (\hat{x}, \hat{y})$ with the semi-classical limit
\begin{align*}
    \lim_{\hbar \to 0} \hat{P}_\hbar (x, y) = P(x, y),
\end{align*}
where the operators $\hat{x}$ and $\hat{y}$ are defined by the following action on some function $f(x)$:
\begin{align*}
    \hat{x} f(x) &= x \cdot f(x), \qquad 
    \hat{y} f(x) = \hbar \frac{d}{dx} f(x).
\end{align*}
The first observation that CEO-TR can be used to derive a quantum curve likely goes back to \cite{Bergere:2009zm}. The main idea is that a wave function $\psi_{x_0}(x)$, which is annihilated by the quantum curve, can be constructed from the family $\{\omega_{g,n}\}$.

\begin{definition}\label{def:wavefunction}
For a given family $\omega_{g,n}$, we define the perturbative wave function by
    \begin{align}\label{pertwave}
    \psi_{x(z_0)}(x(z)) = \exp\bigg[\sum_{\substack{n \geq 1\\ g \geq 0}} \frac{\hbar^{2g + n - 2}}{n!} \int_{z_0}^z \cdots \int_{z_0}^z \bigg(\omega_{g,n}(z_1, \dots, z_n) 
    \\ \notag - \delta_{n,2} \delta_{g,0} \, \frac{dx(z_1) dx(z_2)}{(x(z_1) - x(z_2))^2}\bigg)\bigg],
\end{align}
where $z_0 \in \Sigma$ is a base point lying in the same fundamental domain as $z$.
\end{definition}

As indicated by the definition, the wave function depends on the base point, which leads to different quantizations of the $\hat{P}_\hbar (\hat{x}, \hat{y})$ of the classical spectral curve $P(x, y) = 0$. Bouchard and Eynard proved the following theorem:
\begin{theorem}[\cite{JEP_2017__4__845_0}]\label{thm:quantumcurveBE}
    Let $\{\omega_{g,n}\}$ be generated by CEO-TR. If the underlying spectral curve has genus $0$, then there exists a quantum curve $\hat{P}_\hbar (\hat{x}, \hat{y})$ that annihilates the wave function $\psi_{x(z_0)}(x(z))$, i.e.,
    \begin{align*}
    \hat{P}_\hbar(\hat{x}, \hat{y}) \psi_{x(z_0)}(x(z)) = 0,
    \end{align*}
    where the quantum curve $\hat{P}_\hbar (\hat{x}, \hat{y})$ has only finitely many $\hbar$-corrections and has in the semiclassical limit $\lim_{\hbar \to 0} \hat{P}_\hbar (x, y) $ the polynomial $P(x,y)$ as a factor. 
\end{theorem}

\begin{remark}
	For special choices of $z_0$, in particular, if $x(z_0)$ is singular, the quantum curve simplifies.
\end{remark}

In other words, TR can be used to derive a quantum curve in a canonical way, which depends only on the choice of the integration divisor, the base point. The ordering of the operators $\hat{x}, \hat{y}$ in $\hat{P}_\hbar$ as well as the derivation of the $\hbar$-corrections depend on the choice of the base point. 
From a quantum mechanical perspective, the theorem states that for a genus zero algebraic spectral curve, CEO-TR (implying simple ramification points) generates the WKB expansion. Without explicitly applying the formulas of TR, \cite{JEP_2017__4__845_0} provides an algorithm to obtain the quantum curve that eventually annihilates the wave function.

For generic base point $z_0$, the construction and result is actually more involved. In \cite{Eynard:2017jij} it was shown that if $P(x,y)$ has degree $d$, then quantum curve for generic base point $z_0$ has at most degree $d^2$ and the coefficients are rational function of $x(z)=x$ and $x(z_0)=x_0$. For generic base point, it might also be useful to write the quantum curve not just with the operators $\hat{x}=x$ and $\hat{y}=\hbar\frac{d}{dx}$ but also with $\hat{x}_0=x_0$ and $\hat{y}_0=-\hbar \frac{d}{dx_0}$. One might also write the quantum curve as $\hat{P}_\hbar(\hat{x},\hat{y};\hat{x}_0,\hat{y}_0)$. This discussion highlights the issues of different representations of quantum curves which are everything else than canonical.

\begin{example}
    Let us discuss the Airy spectral curve of Example \ref{ex:KontWitt} from the quantum curve perspective. The wave function for generic base point $x_0=x(z_0)$ is written in terms of the asymptotic of the Airy function \cite{Bergere:2009zm}
    \begin{align*}
        Ai_{\pm \hbar}(x)=\frac{e^{\mp\frac{2}{3\hbar} x^{3/2}}}{\sqrt{2\pi} x^{1/4}}\sum_{k=0}^\infty \frac{(6k)!}{1296^k(2k)!(3k)!}\left(\mp\frac{\hbar}{x^{3/2}}\right)^k
    \end{align*}
    providing the wave function
    \begin{align*}
        \psi_{x_0}(x)=Ai_{+\hbar}(x)Ai'_{-\hbar}(x_0)-Ai'_{+\hbar}(x)Ai_{-\hbar}(x_0).
    \end{align*}
    This wave function is annihilated by the quantum curve \cite{Eynard:2017jij}
    \begin{align*}
       \hat{P}_\hbar(\hat{x},\hat{y};\hat{x}_0,\hat{y}_0)= \hat{y}^2-\hat{x}-\hbar\frac{\hat{y}-\hat{y}_0}{\hat{x}-\hat{x}_0}.
    \end{align*}
\end{example}

\subsubsection{Symplectic transformations}

Quantum curves can be derived from certain transformations of already known quantum curves.
Let us apply for later use one of the symplectic transformations discussed in Property \ref{properties:omega} to a general quantum curve $\hat{P}_\hbar(\hat{x}, \hat{y})$. 
\begin{example}\label{ex:symquatrafo}
Consider a spectral curve $(\Sigma=\mathbb{P}^1,x,y,B=\frac{dz_1dz_2}{(z_1-z_2)^2})$ rational $x$ and $y$. Assume the wave function of Definition \ref{def:wavefunction} is annihilated by a quantum curve $\hat{P}_\hbar(\hat{x}, \hat{y})$ 
as stated in Theorem \ref{thm:quantumcurveBE}.

 If we now transform $(x,y)$ of the spectral curve by 
\begin{align}\label{linearsymplectictrafo}
    (x,y)\qquad \mapsto\qquad (x,y+R(x))
\end{align}
with some rational function $R(x)$, then all $\omega_{g,n}$ generated by CEO-TR don't change for $2g+n-2\geq 0$. However, $\omega_{0,1}=ydx$ is changed to $(y+R(x))dx$. Thus, \eqref{linearsymplectictrafo} transforms the wave function $\psi_{x(z_0)}(x(z))$ to 
\begin{align} \label{eq:wavefunction-after-transformation}
    \exp\bigg(\frac{1}{\hbar}\int_{z_0}^z R(x(z'))dx(z')\bigg)\psi_{x(z_0)}(x(z)).
\end{align}
If the quantum curve $\hat{P}_\hbar(\hat{x},\hat{y})$ annihilates the wave function $\psi_{x(z_0)}(x(z))$, then the wave function~\eqref{eq:wavefunction-after-transformation} is annihilated by 
\begin{align}\label{quantumcurvexylinearsympl}
    \hat{P}_{\hbar}\big(\hat{x},\hat{y}-R(\hat{x})\big).
\end{align}

\begin{remark}
	Note that going from $\hat{P}_{\hbar}\big(\hat{x},\hat{y}\big)$ to $\hat{P}_{\hbar}\big(\hat{x},\hat{y}-R(\hat{x})\big)$, it is literally a replacement of $\hat{y}$ by $\hat{y}-R(\hat{x})$ without any normal-ordering due to noncommutativity of $\hat{x}$ and $\hat{y}$.
\end{remark}

\begin{remark}
If the quantum curve keeps track of the base point and thus has a representation of the form $\hat{P}_{\hbar}\big(\hat{x},\hat{y};\hat{x}_0,\hat{y}_0\big)$, then the transformation $(x,y)\mapsto (x,y+R(x))$ would amount the following transformation of the quantum curve:
\begin{align*}
	\hat{P}_{\hbar}\big(\hat{x},\hat{y};\hat{x}_0,\hat{y}_0\big)\quad \mapsto\quad \hat{P}_{\hbar}\big(\hat{x},\hat{y}-R(\hat{x});\hat{x}_0,\hat{y}_0-R(\hat{x}_0)\big).
\end{align*}	
\end{remark}
\end{example}

From here, there are different ways to attempt to generalize Theorem \ref{thm:quantumcurveBE} to higher genus algebraic spectral curves, to curves of higher order ramification, or to curves including exponentials $e^x$ and $e^y$, which are therefore not algebraic anymore. 

\subsubsection{Higher genus}

For higher genus algebraic spectral curves, the construction of the wave function needs to include non-perturbative effects coming from non-trivial cycles on the curve. The actual construction of the full non-perturbative wave function, i.e., the \textit{exact} WKB expansion, is not directly produced by CEO-TR alone. For the Weierstrass curve, Iwaki provided the construction in \cite{Iwaki:2019zeq}, which was further developed in \cite{Marchal:2019bia,Eynard:2023fil,Eynard:2021sxg}. Certain modularity properties must be satisfied by the exact WKB expansion, which are recursively constructed in those articles for algebraic curves with simple ramification points, thus CEO-TR-type.

In this article we want to deal only with genus zero spectral curves, so Definition \ref{def:wavefunction} for the wave function is sufficient for us.  On the other hand, we are interested in the quantum spectral curve in combination with the different generalizations of TR mentioned in Sec. \ref{Sec:ExtTR}.	

\subsubsection{Non-algebraic curves and Log-TR}

In the generalization to non-algebraic curves including exponentials $e^x$ and $e^y$, it was shown in some case studies that CEO-TR can also quantize those curves related, for instance, to the B-model mirror curve for non-compact toric Calabi-Yau threefolds in Gromov-Witten theory, and/or simple $A$-polynomials \cite{Gukov:2011qp}. Several further examples have been proven in the past related to Hurwitz theory and further Gromov-Witten invariants \cite{Do:2014gda,Mulase:2012tm,Dunin-Barkowski:2013wca}. However, general statements about those curves are not available and have been constructed only on a case-by-case basis. Worse, it has been observed that CEO-TR for non-algebraic curves (for instance, the topological vertex curve \cite{Zhou:2009ea}) \textit{cannot} provide the correct quantization in general (see \cite{Bouchard:2011ya,Gukov:2011qp}). It was unclear under which circumstances CEO-TR with simple ramification points for $x$ but non-meromorphic $x, y$ (though meromorphic $dx$ and $dy$) provides the correct quantum curve and when it does not.

It was observed in \cite{Hock:2023dno} that by combining (and enforcing) the so-called $x$-$y$ duality of TR (which will be reviewed in the subsequent subsection) with specific non-algebraic spectral curves, it is possible to formulate an extension of TR that might provide the correct quantization of non-algebraic curves. This led in~\cite{Alexandrov:2023tgl-LogTR} to a definition of Log-TR surveyed in Section~\ref{sec:logTR}. 

From the definition of Log-TR, we can now understand why, for certain examples of non-algebraic spectral curves, CEO-TR can indeed provide the correct quantization, and in other cases, why it cannot. The reason is that Log-TR reduces to CEO-TR if there are no Log-TR-vital poles, even if there are logarithmic singularities. This was already discussed in Section \ref{sec:logTR}. In all (non-algebraic) examples quantized by CEO-TR in the literature up to now, this was indeed the case. However, the observation that CEO-TR cannot provide a quantization for any non-algebraic curve \cite{Bouchard:2011ya,Gukov:2011qp} in general arises precisely in situations where a Log-TR-vital pole exists. Thus, Log-TR is the ideal candidate for providing the correct tool to quantize non-algebraic curves, including the $A$-polynomial in knot theory related to the AJ-conjecture (where, however, the non-perturbative corrections have to be included).

\subsubsection{Towards Gen-TR}
Furthermore, one may ask whether Gen-TR also gives rise to a quantum curve in general. We have seen that Gen-TR reduces to CEO-TR, CN-TR, BE-TR, and Log-TR, but does it construct a quantum curve in its full generality? More precisely, what happens if the set $\mathcal{P}$ is \textit{not} chosen to coincide with the set of critical points of $x$? Examples in this direction will follow from the main theorem of this article. We find that Gen-TR can produce indeed a quantum curve even if it does not reduce to one of the other versions of TR, specifically when the set $\mathcal{P}$ is not equal to the set of critical points of $x$.

\subsection{Universal dualities}

The last thing that we need to review before moving to the main part of this article is the theory of $x$-$y$ duality, which was the fundamental motivation behind the definitions of Log-TR and Gen-TR.

\subsubsection{Universal $x$-$y$ duality} A fundamental question in the theory of TR is the effect of the swap of $x$ and $y$ in the input data of the recursion. Up to a sign it is also a symplectic transformation, however it doesn't preserve $\{\omega_{g,n}\}$.  

A universal formula for the action of the swap of $x$ and $y$ on the system  of differentials $\{\omega_{g,n}\}$ of CEO-TR (the $x$-$y$ swap formula) was conjectured in~\cite{Borot:2021thu}, simplified in \cite{Hock:2022pbw}, proved in a special case in \cite{Hock:2022wer} and proved in general in \cite{Alexandrov:2022ydc}. This universal formula will not be used explicitly in the rest of the article (it can be found in any of the op.~cit., see also \cite{Hock:2023qii, Alexandrov:2024tjo}) instead, we focus on the theorems that follow from it. 

\begin{theorem}\cite{Alexandrov:2024tjo}\label{thm:TRdualityStatements}
	Let $dx$ and $dy$ be meromorphic differentials on a compact Riemann surface $\Sigma$. Split the set of special points as $\mathcal{P}\uplus \mathcal{P}^\vee$. Then the systems of differentials $\{\omega_{g,n}\}$ and $\{\omega_{g,n}^\vee\}$ associated by Gen-TR to the input data $(\Sigma,dx,dy,B,\mathcal{P})$ and $(\Sigma,dy,dx,B,\mathcal{P}^\vee)$ are related to each other by the $x$-$y$ swap formula.  
\end{theorem}

\begin{remark} The assumptions of Gen-TR can eventually specialize to the assumptions of CEO-TR, CN-TR or Log-TR, or special cases of BE-TR with $y$ either having a simple pole or being regular with $dy$ is non-vanishing at the critical points of $x$ (and the similar special case with the roles of $x$ and $y$ are swapped on the dual side). See \cite{Alexandrov:2022ydc,Alexandrov:2023tgl-LogTR} for the versions of this theorem with more restrictive assumptions.
\end{remark}

Thus $x$-$y$ duality is present in and applicable to the different versions of TR. As mentioned earlier, the $x$-$y$ swap formula was at the roots of the definitions of Log-TR and Gen-TR. 

\begin{remark} 
In fact, the universal $x$-$y$ swap formula can be applied without any assumption of TR for the systems of differentials. 

\begin{definition} We call a system of symmetric meromorphic differential $\{\omega_{g,n}\}$ \emph{admissible} if all of them but $\omega_{0,2}$ are regular on all diagonals, and $\omega_{0,2}$ has a pole of order two on the diagonal with biresidue $1$. 
\end{definition}

We don't assume that this system of differentials is coming from TR. In this setting the universal $x$-$y$ swap formula applied to an admissible system $\{\omega_{g,n}\}$ produces an admissible system $\{\omega^\vee_{g,n}\}$, and its $x$-$y$ swap with the roles of $x$ and $y$ interchanged is the initial system  $\{\omega_{g,n}\}$.
\end{remark}

\subsubsection{Universal symplectic duality}\label{sec:symduality}

The $x$-$y$ duality is a specific symplectic transformation of the formal symplectic form $dx \wedge dy$. Combining it with the basic transformation $(x,y)\to (x,y+R(x))$ discussed in Properties \ref{properties:omega} and Example~\ref{ex:symquatrafo}, we obtain
via the chain of transformations
\begin{align*}
	\begin{matrix}
		(x, y) & \longrightarrow & (y, x) \\
		\rotatebox{-90}{$\dashrightarrow$} &  & \rotatebox{-90}{$\longrightarrow$} \\
		\\
		(x+R(y), y) & \longleftarrow & (y, x+R(y))
	\end{matrix}
\end{align*}
the following one:
\begin{align}\label{symplecticduality}
    (x,y) \quad &\longrightarrow \quad (x+R(y), y).
\end{align}
It is studied in \cite{Bychkov:2022wgw, Alexandrov:2024ajj} under the name \emph{symplectic duality}. Since it is a composition of transformations that can be represented by universal explicit formulas, its action on a system of differentials is also given by an explicit universal formula, see~\cite[Definition 4.1]{Alexandrov:2024ajj}. Note that the function $R(y)$ is allowed to be a series in $\hbar^2$.

\subsubsection{Universal duality for kernels}
For a given system of admissible differentials $\omega_{g,n}$, define the kernel as
\begin{align}\label{kernel}
K(z_1,z_2) = \frac{\exp\bigg(\sum\limits_{(g,n)\neq (0,1)}\frac{\hbar^{2g+n-2}}{n!}\int\limits_{z_2}^{z_1}...\int\limits_{z_2}^{z_1} \omega_{g,n} - \frac{\delta_{(g,n),(0,2)} dx(z_1)dx(z_2)}{(x(z_1)-x(z_2))^2}\bigg)}{x(z_1) - x(z_2)},
\end{align}
where $z_1,z_2$ are in the same fundamental domain. By definition, $K(z_1,z_2)$ is a formal power series expansion in $\hbar$. 

\begin{remark} In principle, one gets a more invariant definition if $K$ is defined to be a bi-half-differential. To this end, the right hand side of  Equation~\eqref{kernel} is typically multiplied by $\sqrt{dx(z_1)}\sqrt{dx(z_2)}$. 
\end{remark}

One can apply the universal $x$-$y$ swap formula to $\{\omega_{g,n}\}$ and obtain a new system of differentials $\{\omega^\vee_{g,n}\}$. Let $K^\vee(z_1,z_2)$ be the kernel associated to  $\{\omega^\vee_{g,n}\}$.

\begin{theorem}[\cite{Alexandrov:2024tjo}]\label{thm:kernels}
We have
   \begin{align}\label{kernelduality2}
        K^\vee(z_1,z_2) = \frac{i}{2\pi \hbar}
        \iint&K(\chi_1,\chi_2) dx(\chi_1) dx(\chi_2) \\\nonumber
        &\times e^{\frac{1}{\hbar}\big(y(z_2) (x(z_2)-x(\chi_1))+\int_{z_2}^{\chi_1} y\, dx\big)} \\\nonumber&\times e^{-\frac{1}{\hbar}\big(y(z_1) (x(z_1)-x(\chi_2))+\int_{z_1}^{\chi_2} y\, dx\big)},
    \end{align}
    where $\int$
    is understood as formal Gaussian integrals for nondegenerate quadratic forms at each order in $\sqrt{\hbar}$. The dual transformation reconstructs $K(z_1,z_2)$ from $K^\vee(z_1,z_2)$ after interchanging $x$ and $y$
     \begin{align}\label{kernelduality}
        K(z_1,z_2) = \frac{i}{2\pi \hbar}
        \iint&K^\vee(\chi_1,\chi_2) dy(\chi_1) dy(\chi_2) \\\nonumber
        &\times e^{\frac{1}{\hbar}\big(x(z_2) (y(z_2)-y(\chi_1))+\int_{z_2}^{\chi_1} x\, dy\big)} \\\nonumber
        &\times e^{-\frac{1}{\hbar}\big(x(z_1) (y(z_1)-y(\chi_2))+\int_{z_1}^{\chi_2} x\, dy\big)}.
    \end{align}
\end{theorem}

\begin{remark} Note the different sign conventions compared to \emph{op.~cit.}. We use $\omega_{0,1}=ydx$ (as opposed to $\omega_{0,1}=-ydx$), so we have to replace $\hbar$ with $-\hbar$.
\end{remark}

\begin{remark}
The kernel $K$, as defined in \eqref{kernel}, is related to the wave function of Definition \ref{def:wavefunction} via:
\begin{align}
	e^{-\frac{1}{\hbar}\int_{z_0}^z\omega_{0,1}} \frac{\psi_{x(z_0)}(x(z))}{x(z)-x(z_0)} = K(z,z_0),
\end{align}
where $\omega_{0,1} = y\,dx$.	
\end{remark}

\section{Quantum curve, wave function, and dualities}

\label{sec:QuantumCurve-dualities}

This section aims to combine the ideas of constructing the wave function and the corresponding quantum curve from the different versions of TR, focusing on the $x$-$y$ and symplectic dualities. The first computations attempting to combine the concepts of the quantum curve and the $x$-$y$ duality were carried out in \cite{Hock:2023dno}. This idea was further studied in \cite{Weller:2024msm} for spectral curves with singular base points for $x$ and $y$. More precisely, the spectral curves studied in \cite{Weller:2024msm} require a point $z_0$ such that $dx$ and $dy$ are singular at this point. In this setting, the two perturbative wave functions $\psi_{x(z_0)}(x(z))$ and $\psi_{y(z_0)}^\vee(y(z))$ are related via formal Laplace transform. 

General spectral curves do not need to satisfy the property that a point $z_0$ exists where both $dx$ and $dy$ are singular. One may ask whether it is possible to relate the two perturbative wave functions $\psi_{x(z_0)}(x(z))$ and $\psi^\vee_{y(\tilde{z}_0)}(y(z))$ with different base points $z_0$ and $\tilde{z}_0$. The simple answer to this question is yes, but the result is somewhat unsatisfactory, as it is essentially a reformulation of Theorem~\ref{thm:kernels}.

\begin{proposition}\label{prop:wavefunctionduality}
Let  $\{\omega_{g,n}\}$ be an admissible system of differentials and let $\{\omega_{g,n}^\vee\}$ be its universal $x$-$y$ swap. 
Let the perturbative wave function $\psi_{x(z_0)}(x(z))$ (resp., $\psi^\vee_{y(\tilde{z}_0)}(y(z))$)  be defined from $\{\omega_{g,n}\}$ (resp., $\{\omega_{g,n}^\vee\}$) with $\omega_{0,1} = y\,dx$ (resp., $\omega_{0,1}^\vee = x\,dy$). Then, the following relation holds:
\begin{align}\label{ExtLaplace}
    \frac{\psi_{x(z_0)}(x(z))}{x(z)-x(z_0)} = \frac{i}{2\pi\hbar} \iint \frac{dy(\chi)dy(\chi_0)}{y(\chi)-y(\chi_0)} e^{\frac{x(z)y(\chi_0)-x(z_0)y(\chi)}{\hbar}} \psi^\vee_{y(\chi_0)}(y(\chi)).
\end{align}
\end{proposition}

\begin{proof}
Recall Equation~\eqref{kernelduality}. Replace the kernels by the wave functions using
\begin{align*}
    K(z,z_0) &= e^{-\frac{1}{\hbar}\int_{z_0}^z y\,dx} \frac{\psi_{x(z_0)}(x(z))}{x(z)-x(z_0)}, \\
    K^\vee(\chi,\chi_0) &= e^{-\frac{1}{\hbar}\int_{\chi_0}^\chi x\,dy} \frac{\psi^\vee_{y(\chi_0)}(y(\chi))}{y(\chi)-y(\chi_0)}.
\end{align*}
The duality statement of the kernels $K$ and $K^\vee$ is a statement order by order in an $\hbar$-expansion. We can multiply, without any harm, by functions of the form $e^{\frac{f(z,z_0)}{\hbar}}$ such that the statement remains valid. The Gaussian integrals should still be understood as formal Gaussian integrals, even though some formal cancellations may appear.
Accepting this formal manipulation of the duality formula \eqref{kernelduality2}, and multiplying it by $e^{\frac{1}{\hbar}\int_{z_0}^z y\,dx}$, we are left with:
\begin{align*}
    \frac{\psi_{x(z_0)}(x(z))}{x(z)-x(z_0)} 
    =\frac{i}{2\pi\hbar} \iint & \frac{dy(\chi)dy(\chi_0)}{y(\chi)-y(\chi_0)} 
    e^{\frac{1}{\hbar}\int_{z_0}^z y\,dx} e^{-\frac{1}{\hbar}\int_{\chi_0}^\chi x\,dy} \psi^\vee_{y(\chi_0)}(y(\chi)) \\
    &\times e^{-\frac{1}{\hbar}\big(x(z)(y(z)-y(\chi_0))+\int_{z}^{\chi_0} x\,dy\big)} \\
     &\times e^{\frac{1}{\hbar}\big(x(z_0)(y(z_0)-y(\chi))+\int_{z_0}^{\chi} x\,dy\big)} .
\end{align*}
The argument of the exponent under the sign of integral is reduced by the one in~\eqref{ExtLaplace} by partial integration. 
\end{proof}

\begin{remark} In Equation~\eqref{ExtLaplace}, the Laplace kernel has the form $e^{\frac{1}{\hbar}x(z)y(\chi_0)}$. Thus, it relates the base point of $\psi^\vee$ with the actual variable of $\psi$. Note, however, that by the definition of the wave function $\psi_{x_0}(x)$ with some formal parameter $\hbar$, changing the sign of $\hbar$ amounts to swapping $x_0$ and $x$: 
		$\psi^\hbar_{x_0}(x) = \psi^{-\hbar}_{x}(x_0)$.
	Thus, changing the sign of $\hbar$ for the dual wave function  is a more natural choice. We retain the definitions as given earlier.
\end{remark}

If the multi-differentials $\omega_{g,n}$ originate from one of the various versions of topological recursion, we conclude the following corollary:

\begin{corollary}
    If $\omega_{g,n}$ is generated by CEO-TR, CN-TR, BE-TR, or Log-TR for the spectral curve $(\Sigma,dx,dy,B)$, or by Gen-TR for the spectral curve $(\Sigma,dx,dy,B,\mathcal{P})$ under the conditions specified in Theorem \ref{thm:TRdualityStatements}, and $\omega_{g,n}^\vee$ is generated by the dual spectral curve $(\Sigma,dy,dx,B)$ or $(\Sigma,dy,dx,B,\mathcal{P}^\vee)$, respectively, then the two perturbative wave functions are related via \eqref{ExtLaplace}.
\end{corollary}

\begin{remark}
    The extended Laplace transform \eqref{ExtLaplace} of Proposition \ref{prop:wavefunctionduality}, when inserted into its dual version, formally yields the identity.
\end{remark}

\subsection{\texorpdfstring{$x$-$y$}{x-y} duality transformation for the quantum curve}\label{sec:class}
Proposition \ref{prop:wavefunctionduality} indicates that the base point is indeed important from several perspectives. If the wave functions $\psi$ and $\psi^\vee$ are derived, for instance, from CEO-TR, the theorem by Bouchard-Eynard (Theorem \ref{thm:quantumcurveBE}) guarantees the existence of a quantum curve $\hat{P}_{\hbar}(\hat{x},\hat{y})$ that annihilates the wave function. Thus, it becomes interesting to study the action of $\hat{x}$ and $\hat{y}$ on the duality formula of the wave functions \eqref{ExtLaplace}. This provides a way to connect the quantum curve $\hat{P}_\hbar(\hat{x},\hat{y})$ to its dual $\hat{P}^\vee_\hbar(\hat{x}^\vee,\hat{y}^\vee)$. 

This approach can be effectively used if one of the families $\omega_{g,n}$ or $\omega_{g,n}^{\vee}$ is trivial, enabling us to derive the quantum curve of the other family. To achieve this, let us compute the action of derivatives with respect to $x(z)$ and multiplication by $x(z)$ on \eqref{ExtLaplace}. The following proposition derives, for the action of $\hat{x}$ and $\hat{y}$ on $\psi_{x_0}(x)$, the corresponding dual action on $\psi^\vee_{y_0}(y)$. The operators, including those for the base points $x_0$ and $y_0$, are defined as:
\begin{align*}
    &\hat{x}f(x,x_0) = x \cdot f(x,x_0), 
    &&\hat{y} f(x,x_0) = \hbar \frac{d}{dx} f(x,x_0), \\
    &\hat{x}_0f(x,x_0) = x_0 \cdot f(x,x_0), 
    &&\hat{y}_0 f(x,x_0) = -\hbar \frac{d}{dx_0} f(x,x_0),
\end{align*}
and, for the dual side, we have 

\begin{align*}
	&\hat{x}^\vee f(x^\vee,x^\vee_0) = x^\vee \cdot f(x^\vee,x^\vee_0), 
	&&\hat y^\vee f(x^\vee,x^\vee_0) = \hbar \frac{d}{dx^\vee} f(x^\vee,x^\vee_0), \\
	&\hat x^\vee_0f(x^\vee,x^\vee_0) = x^\vee_0 \cdot f(x^\vee,x^\vee_0), 
	&&\hat y^\vee_0 f(x^\vee,x^\vee_0) = -\hbar \frac{d}{dx^\vee_0} f(x^\vee,x^\vee_0),
\end{align*}

\begin{proposition} \label{prop:ActionOnOperators}
    In application to the two wave functions related by \eqref{ExtLaplace}, the operators $\hat{x},\hat{x}_0,\hat{y},\hat{y}_0$ transform to expressions in terms of the dual operators $\hat{x}^\vee,\hat{x}_0^\vee,\hat{y}^\vee,\hat{y}_0^\vee$ as follows:
    \begin{align*}
        &\hat{x} \quad \mapsto \quad \hat{y}_0^\vee - \frac{\hbar}{\hat{x}^\vee - \hat{x}_0^\vee},\qquad 
        &&\hat{y} \quad \mapsto \quad \hat{x}_0^\vee - \frac{\hbar}{\hat{y}^\vee - \hat{y}_0^\vee},\\
        &\hat{x}_0 \quad \mapsto \quad \hat{y}^\vee - \frac{\hbar}{\hat{x}^\vee - \hat{x}_0^\vee},\qquad 
        &&\hat{y}_0 \quad \mapsto \quad \hat{x}^\vee - \frac{\hbar}{\hat{y}^\vee - \hat{y}_0^\vee}.
    \end{align*}
\end{proposition}
\begin{proof} Differentiate~\eqref{ExtLaplace} with respect to $\hbar \frac{d}{dx(z)}$. On the one hand, we have 
\begin{align*}
	&\hbar \frac{d}{dx(z)} \frac{\psi_{x(z_0)}(x(z))}{x(z)-x(z_0)} = \frac{\big(\hbar \frac{d}{dx} - \frac{\hbar}{x-x_0}\big)\psi_{x(z_0)}(x(z))}{x(z)-x(z_0)} 
\end{align*}
On the other hand, 
\begin{align*}
   & \hbar \frac{d}{dx(z)} \frac{\psi_{x(z_0)}(x(z))}{x(z)-x(z_0)} \\
   & = \frac{i}{2\pi\hbar}\iint \hbar \frac{d}{dx(z)}\frac{\,dy(\chi)\,dy(\chi_0)}{y(\chi)-y(\chi_0)} e^{\frac{x(z)y(\chi_0)-x(z_0)y(\chi)}{\hbar}} \psi^\vee_{y(\chi_0)}(y(\chi))\\
   & = \frac{i}{2\pi\hbar}\iint \hbar \frac{d}{dx(z)}\frac{\,dy(\chi)\,dy(\chi_0)}{y(\chi)-y(\chi_0)} e^{\frac{x(z)y(\chi_0)-x(z_0)y(\chi)}{\hbar}} y(\chi_0) \psi^\vee_{y(\chi_0)}(y(\chi)).
\end{align*}
Thus, we find:
\begin{align}\label{dxy}
    \hbar \frac{d}{dx} - \frac{\hbar}{x-x_0} \mapsto y_0 = \hat x_0^\vee.
\end{align}
Similarly, 
\begin{align}\label{xdy}
    x \mapsto - \hbar \frac{d}{dy_0} - \frac{\hbar}{y-y_0} = \hat y^\vee_0 - \frac{\hbar}{\hat x^\vee - \hat x^\vee_0}.
\end{align}
For the $x \leftrightarrow x_0$ case, we have
\begin{align}\label{dx0y}
    \hbar \frac{d}{dx_0} + \frac{\hbar}{x-x_0} \mapsto -y = -\hat x^\vee
\end{align}
and
\begin{align}\label{x0dy}
    x_0 \mapsto \hbar \frac{d}{dy} - \frac{\hbar}{y-y_0} = \hat y^\vee - \frac{\hbar}{\hat x^\vee-\hat x^\vee_0}.
\end{align}
Combining \eqref{x0dy} with \eqref{xdy}, we can map the term $\frac{\hbar}{x-x_0}$ in \eqref{dxy} to the other side as well, and we obtain:
\begin{align}\label{dxytofully}
    \hat y = \hbar \frac{d}{dx} \mapsto y_0 - \frac{1}{\frac{d}{dy} + \frac{d}{dy_0}} = \hat x_0^\vee - \frac{\hbar}{\hat y^\vee -\hat y^\vee_0}.
\end{align}
The other equations are proved analogously. 
\end{proof}

\begin{remark} Note that applying the transformations of Proposition~\ref{prop:ActionOnOperators} twice recovers the initial operators. Thus, this action is an involution. 
\end{remark}

\begin{corollary}\label{cor:dualquantumcurve}
    Let us assume $\hat{P}_\hbar(\hat{x},\hat{y};\hat{x}_0,\hat{y}_0)$ is the quantum curve that annihilates the wave function $\psi_{x(z_0)}(x(z))$ generated by the system  of differentials $\{\omega_{g,n}\}$. Then the dual quantum curve $\hat{P}^\vee_\hbar(\hat{x}^\vee,\hat{y}^\vee;\hat{x}_0^\vee,\hat{y}_0^\vee)$ annihilating the dual wave function $\psi^\vee_{x^\vee(z_0)}(x^\vee(z))$ generated by the dual system $\{\omega_{g,n}^\vee\}$ is given by:
    \begin{align}
        & \hat{P}^\vee_\hbar(\hat{x}^\vee,\hat{y}^\vee;\hat{x}_0^\vee,\hat{y}_0^\vee)
        \\
        \notag 
        & = \hat{P}_\hbar\bigg(\hat{y}_0^\vee - \frac{\hbar}{\hat{x}^\vee - \hat{x}_0^\vee},
         \hat{x}_0^\vee - \frac{\hbar}{\hat{y}^\vee - \hat{y}_0^\vee};
        \hat{y}^\vee - \frac{\hbar}{\hat{x}^\vee - \hat{x}_0^\vee},
        \hat{x}^\vee - \frac{\hbar}{\hat{y}^\vee - \hat{y}_0^\vee}\bigg),
    \end{align}
    where the right-hand side is literally understood as substitution (no normal ordering is applied).
\end{corollary}

Interestingly, we achieve a representation of quantum curves with operators $\hat{y}^\vee-\hat{y}_0^\vee=\hbar\frac{d}{dx^\vee}+\hbar\frac{d}{dx^\vee_0}$ in the denominator. By acting with $(\hat{y}^\vee-\hat{y}_0^\vee)^n$ for some $n$ from the left on the quantum curve, these denominators can be successively canceled. However, the commutation relations $[\hat{x}^\vee,\hat{y}^\vee]=-\hbar$,  $[\hat{x}_0^\vee,\hat{y}_0^\vee]=\hbar$ need to be applied. Thus, writing a quantum curve as a polynomial rather than a rational expression in $\hat{x}$ and $\hat{y}$ does not necessarily correspond to the most natural representation.	

For instance, if the family $\omega_{g,n}$ is trivial, i.e., all $\omega_{g,n}=0$ for $2g+n-2>0$ or in the context of Log-TR for all $n>1$, then the quantum curve $\hat{P}_\hbar$ is also trivially obtained. Applying Corollary \ref{cor:dualquantumcurve}, we can easily derive the dual quantum curve $\hat{P}^\vee_\hbar$ for a general base point, which is, in general, already a very involved problem. We will rederive several known quantum curves and also derive new ones through this approach in the subsequent section.

In the literature on quantum curves arising from TR, the base points $(x_0, y_0)$ are usually chosen to have specific values. This choice simplifies the quantum curve, and was first properly analyzed in \cite{JEP_2017__4__845_0}. The first two terms contributing to the wave function, corresponding to $(g,n)=\{(0,1),(0,2)\}$, need to be regularized at singular base points, which amounts to multiplying the wave function by a constant. Tracking general base points has been studied, for instance, in \cite{Eynard:2017jij}, but was found to be quite challenging.

From the perspective of $x$-$y$ duality together with the quantum curve, it was shown in \cite{Weller:2024msm} that the extended Laplace transform of Proposition~\ref{prop:wavefunctionduality} reduces to an ordinary Laplace transform if both base points $x_0$ and $y_0$ are singular. However, not every spectral curve admits such a point where both $x$ and $y$ are simultaneously singular.

Since a simple representation of the quantum curve is achieved by choosing $x_0$ to be a singular point, the following lemma provides the limiting behavior of the wave function—and thus of the quantum curve—in this limit, for general $y_0$.

\begin{lemma}\label{lem:singularpoint}
Let $z_0 \in \Sigma$ be a singular point of $x$, that is $x_0=x(z_0)=\infty$ (including possibly a logarithmic singular point) that is not a ramification point. For Gen-TR, assume that $z_0$ is not in $\mathcal{P}$. Let the wave function $\psi_{x(\tilde{z})}(x(z))$ be regularized (if necessary) by a constant $C_\hbar(z_0)$ such that the limit
\[
\lim_{\tilde{z} \to z_0} C_\hbar(z_0)\, \psi_{x(\tilde{z})}(x(z)) = \psi^{\mathrm{reg}}_{\infty}(x(z))
\]
is regular. Then, for any function $f$ analytic in some ring that contains $y_0$, the following identities hold:
\begin{align*}
\lim_{\tilde{z} \to z_0} C_\hbar(z_0)\, f(x(\tilde{z}))\, \psi_{x(\tilde{z})}(x(z)) &= f(x_0)\, \psi^{\mathrm{reg}}_{\infty}(x(z)), \\
\lim_{\tilde{z} \to z_0} C_\hbar(z_0)\, f\big(-\hbar \tfrac{d}{dx(\tilde{z})} \big)\, \psi_{x(\tilde{z})}(x(z)) &= f(y_0)\, \psi^{\mathrm{reg}}_{\infty}(x(z)).
\end{align*}

\begin{proof}
The first identity is immediate. For the second, let $f$ be of the form $f(t) = \sum_\ell a_\ell t^\ell$. The action of $\left(-\hbar \tfrac{d}{dx(\tilde{z})}\right)^\ell$ on $\psi_{x(\tilde{z})}(x(z))$ produces an expression of the form
\begin{align}\label{lemeq1}
\left( y_0^\ell + \mathcal{O}(\hbar)\, \mathcal{O}((x'(\tilde{z}))^{-1}) \right) \psi_{x(\tilde{z})}(x(z)),
\end{align}
where each factor of $y_0$ arises from the action on the exponential $e^{\frac{1}{\hbar} \int_{\tilde{z}}^z y\, dx}$ inside $\psi_{x(\tilde{z})}(x(z))$. All $\omega_{g,n}$ are regular at $z_0$ since it is not a ramification point. Also in the case of Log-TR or Gen-TR, $z_0$ is not a pole of $\omega_{g,n}$. Multiplying \eqref{lemeq1} by $C_\hbar(z_0)$, the limit reduces to $y_0^\ell \psi^{\mathrm{reg}}_{\infty}(x(z))$ because all higher-order $\hbar$-terms are accompanied by factors of $\frac{1}{x'(\tilde{z})}$, which vanish in the limit $\tilde{z} \to z_0$. Analyticity of $f$ implies convergence of the corrsponding sum over $\ell$. 
\end{proof}
\end{lemma}

\subsection{Symplectic duality transformation for the quantum curve}
As discussed before, symplectic duality, as introduced in \cite{Bychkov:2022wgw,Alexandrov:2023oov}, \cite{Alexandrov:2024ajj}, plays a fundamental role in the application of TR in B-model topological string theory, particularly in conjunction with the BKMP-conjecture. However, this was not studied yet in details. In the application of TR for B-model computations, it is essential to choose a so-called framing \cite{Bouchard:2007ys}. As already mentioned in Sec. \ref{sec:logTR}, the choice of a framing is equivalent to a specific symplectic duality transformation. Since the quantum curve is of considerable interest in topological string theory, Chern-Simons theory, and knot theory, it is natural to ask how the quantum curve transforms under the symplectic duality transformation, particularly for a generic base point.

The symplectic duality induces a duality between two families, $\omega_{g,n}$ and $\omega_{g,n}^\dagger$ (see \cite{Alexandrov:2024ajj}), under a transformation of $(x,y)$ in the framework of topological recursion given by:
\begin{align}\label{symplectictrafo}
    (x,y)\quad \mapsto\quad (x+R(y),y),
\end{align}
for some rational function $R$. The symplectic duality has been proven to hold for CEO-TR and for a very large class of cases in Log-TR.

Let $\hat{P}_\hbar$ be the quantum curve that annihilates the wave function $\psi_{x(z_0)}(x(z))$ constructed from the family $\omega_{g,n}$, and let $\hat{P}^\dagger_\hbar$ be the symplectic dual quantum curve annihilating the symplectic dual wave function $\psi^\dagger_{x^\dagger(z_0)}(x^\dagger(z))$ constructed from the symplectic dual family $\omega_{g,n}^\dagger$. The following proposition relates these two quantum curves:

\begin{proposition}\label{prop:sympquacurve}
    Let $\omega_{g,n},\psi_{x(z_0)}(x(z)), \hat{P}_\hbar$ and $\omega_{g,n}^\dagger,\psi^\dagger_{x^\dagger(z_0)}(x^\dagger(z)), \hat{P}_\hbar^\dagger$ be as above. Then, the two quantum curves are related via
     \begin{align}
    	&\hat{P}^\dagger_{\hbar}(\hat{x}^\dagger,\hat{y}^\dagger;\hat{x}_{0}^\dagger,\hat{y}_{0}^\dagger)
    	\\ \notag 
    	& = \hat{P}_\hbar\bigg(
    	\hat{x}^\dagger -R(	\hat{y}^\dagger - \frac{\hbar}{\hat{x}^\dagger - \hat{x}_{0}^\dagger}),
    	\\ \notag 
    	& \qquad \quad
    	\hat{y}^\dagger - \frac{\hbar}{\hat{x}^\dagger - \hat{x}_{0}^\dagger}
    	+  \frac{\hbar}{\hat{x}^\dagger -R(\hat{y}^\dagger - \frac{\hbar}{\hat{x}^\dagger - \hat{x}_{0}^\dagger}) -\hat{x}_{0}^\dagger +R(\hat{y}_{0}^\dagger - \frac{\hbar}{\hat{x}^\dagger- \hat{x}_{0}^\dagger}) },
    	\\ \notag 
    	& \qquad \quad
    	\hat{x}_{0}^\dagger -R(\hat{y}_{0}^\dagger - \frac{\hbar}{\hat{x}^\dagger- \hat{x}_{0}^\dagger}) ,
    	\\ \notag 
    	& \qquad \quad
    	\hat{y}_{0}^\dagger - \frac{\hbar}{\hat{x}^\dagger- \hat{x}_{0}^\dagger} + \frac{\hbar}{\hat{x}^\dagger -R(\hat{y}^\dagger - \frac{\hbar}{\hat{x}^\dagger - \hat{x}_{0}^\dagger}) -\hat{x}_{0}^\dagger +R(\hat{y}_{0}^\dagger - \frac{\hbar}{\hat{x}^\dagger- \hat{x}_{0}^\dagger}) }\bigg).
    \end{align}
\end{proposition}
\begin{proof}
        The proof is carried out in three steps. To realize the symplectic duality transformation, we follow the sequence of transformations as in Section~\ref{sec:symduality}:
        \begin{itemize}
            \item[1.] Apply $x$-$y$ duality $(x,y)\mapsto (x_1,y_1) = (y,x)$.
            \item[2.] Apply the symplectic transformation $(x_1,y_1)\mapsto  (x_2,y_2) = (x_1,y_1+R(y_1))$.
            \item[3.] Apply $x$-$y$ duality again $(x_2, y_2)\mapsto (x_3,y_3)=(y_2,x_2)$.
        \end{itemize}
        The transformation of the quantum curve in Step 2 was already discussed in Example \ref{ex:symquatrafo}. The transformation of the quantum curve under the $x$-$y$ duality follows from Corollary \ref{cor:dualquantumcurve}. Applying all three transformations to the quantum curve one after the other yields the desired result. We will call the quantum curve obtained after the first transformation $\hat{P}_1$, after the second transformation $\hat{P}_2$, and after the third transformation $\hat{P}_3=\hat{P}^\dagger$. 

        Following Corollary \ref{cor:dualquantumcurve}, we have
        \begin{align*}
        	& \hat{P}_{1,\hbar}(\hat{x}_1,\hat{y}_1;\hat{x}_{1,0},\hat{y}_{1,0})
        	\\
        	\notag 
        	& = \hat{P}_\hbar\Big(\hat{y}_{1,0} - \tfrac{\hbar}{\hat{x}_1- \hat{x}_{1,0}},
        	\hat{x}_{1,0} - \tfrac{\hbar}{\hat{y}_1 - \hat{y}_{1,0}};
        	\hat{y}_1 - \tfrac{\hbar}{\hat{x}_1 - \hat{x}_{1,0}},
        	\hat{x}_1 - \tfrac{\hbar}{\hat{y}_1 - \hat{y}_{1,0}}\Big),
        \end{align*}
         The transformation $(x_1,y_1)\mapsto  (x_2,y_2) = (x_1,y_1+R(y_1))$ yields
        \begin{align*}
        	&\hat{P}_{2,\hbar}(\hat{x}_2,\hat{y}_2;\hat{x}_{2,0},\hat{y}_{2,0})
\\
& =\hat{P}_{1,\hbar}\big(\hat x_2, \hat{y}_2-R(\hat{x}_2);\hat x_{2,0}, \hat{y}_{2,0}-R(\hat{x}_{2,0})\big)
        	\\
        	& 	= \hat{P}_\hbar\Big(\hat{y}_{2,0}-R(\hat{x}_{2,0}) - \tfrac{\hbar}{\hat{x}_2- \hat{x}_{2,0}},
        	\hat{x}_{2,0} - \tfrac{\hbar}{\hat{y}_2-R(\hat{x}_2) - \hat{y}_{2,0}+R(\hat{x}_{2,0})}
        	\\
        	& \qquad \quad
        	\hat{y}_2-R(\hat{x}_2) - \tfrac{\hbar}{\hat{x}_2 - \hat{x}_{2,0}},
        	\hat{x}_2 - \tfrac{\hbar}{\hat{y}_2-R(\hat{x}_2) -\hat{y}_{2,0}+R(\hat{x}_{2,0})}\Big).
        \end{align*}
                The third transformation leads to 
        \begin{align*}
        	&\hat{P}_{3,\hbar}(\hat{x}_3,\hat{y}_3;\hat{x}_{3,0},\hat{y}_{3,0})\\
        	& = \hat{P}_{2,\hbar}\Big(\hat{y}_{3,0} - \tfrac{\hbar}{\hat{x}_3- \hat{x}_{3,0}},
        	\hat{x}_{3,0} - \tfrac{\hbar}{\hat{y}_3 - \hat{y}_{3,0}};
        	\hat{y}_3 - \tfrac{\hbar}{\hat{x}_3 - \hat{x}_{3,0}},
        	\hat{x}_3 - \tfrac{\hbar}{\hat{y}_3 - \hat{y}_{3,0}}\Big),
        	\\
        	& = \hat{P}_\hbar\Big(\hat{x}_3 - \tfrac{\hbar}{\hat{y}_3 - \hat{y}_{3,0}}-R(	\hat{y}_3 - \tfrac{\hbar}{\hat{x}_3 - \hat{x}_{3,0}}) - \tfrac{\hbar}{\hat{y}_{3,0} - \tfrac{\hbar}{\hat{x}_3- \hat{x}_{3,0}}- 	\hat{y}_3 + \tfrac{\hbar}{\hat{x}_3 - \hat{x}_{3,0}}},
        	\\
        	& \qquad \quad
        		\hat{y}_3 - \tfrac{\hbar}{\hat{x}_3 - \hat{x}_{3,0}} - \tfrac{\hbar}{\hat{x}_{3,0} - \tfrac{\hbar}{\hat{y}_3 - \hat{y}_{3,0}}-R(\hat{y}_{3,0} - \tfrac{\hbar}{\hat{x}_3- \hat{x}_{3,0}}) - \hat{x}_3 + \tfrac{\hbar}{\hat{y}_3 - \hat{y}_{3,0}}+R(\hat{y}_3 - \tfrac{\hbar}{\hat{x}_3 - \hat{x}_{3,0}})}
        	\\
        	& \qquad \quad
        	\hat{x}_{3,0} - \tfrac{\hbar}{\hat{y}_3 - \hat{y}_{3,0}}-R(\hat{y}_{3,0} - \tfrac{\hbar}{\hat{x}_3- \hat{x}_{3,0}}) - \tfrac{\hbar}{\hat{y}_{3,0} - \tfrac{\hbar}{\hat{x}_3- \hat{x}_{3,0}} - 	\hat{y}_3 + \tfrac{\hbar}{\hat{x}_3 - \hat{x}_{3,0}}},
        	\\
        	& \qquad \quad
        	\hat{y}_{3,0} - \tfrac{\hbar}{\hat{x}_3- \hat{x}_{3,0}} - \tfrac{\hbar}{\hat{x}_{3,0} - \tfrac{\hbar}{\hat{y}_3 - \hat{y}_{3,0}}-R(\hat{y}_{3,0} - \tfrac{\hbar}{\hat{x}_3- \hat{x}_{3,0}}) -\hat{x}_3 + \tfrac{\hbar}{\hat{y}_3 - \hat{y}_{3,0}}+R(\hat{y}_3 - \tfrac{\hbar}{\hat{x}_3 - \hat{x}_{3,0}})}\Big),
        \end{align*}
        which yields the final result.
    \end{proof}

\section{List of examples}

\label{sec:examples}

The examples will be organized according to the corresponding version of TR. In the literature, many quantum curves have already been derived for enumerative problems related to Hurwitz theory. These results are obtained using CEO-TR, because, from the perspective of Log-TR, there are no Log-TR-vital singularities. Thus, Log-TR reduces to CEO-TR. However, the dual side does exhibit Log-TR-vital singular points. This implies that these examples inherently require the framework of Log-TR to be fully described.

\subsection{CEO-TR-like examples}
\subsubsection{Rational function example via $x$-$y$ duality}
We derive a quantization of the spectral curve (assuming the setting of CEO-TR) 
\begin{align} \label{eq:Example-Rational-x}
\Big(\Sigma=\mathbb{P}^1,x=\frac{p(z)}{q(z)},y=z,B=\frac{dz_1\,dz_2}{(z_1-z_2)^2}\Big)
\end{align}
via $x$-$y$ duality.
The functions $x$ and $y$ satisfy
\begin{align}\label{pqcurve}
    P(x,y)= p(y)-q(y)x=0
\end{align}
where we assume that $p(y)$ and $q(y)$ are coprime polynomials, hence, the curve is irreducible. 

The dual family $\omega_{g,n}^\vee$ is for the curve \eqref{pqcurve} trivial, that is 
\begin{align*}
    &\omega_{0,1}^\vee= x(z)\,dz,\quad \omega_{0,2}^\vee=\frac{dz_1\,dz_2}{(z_1-z_2)^2},\\
    &\omega_{g,n}^\vee=0\quad \text{for $2g+n-2>0$}.
\end{align*}
The dual wave function wave function $\psi^\vee_{x^\vee(z_0)}(x^\vee(z))$ factorizes in two parts depending on $x^\vee(z_0)$ and $x^\vee(z)$ since $\omega_{0,1}^\vee$ gives the only contribution. This means that the dual quantum curve can just depend on $\hat{x}_0^\vee$ and $\hat{y}_0^\vee$. The dual wave function $\exp({\frac{1}{\hbar}\int_{x^\vee(z_0)}^{x^\vee(z)} \frac{p(\zeta)}{q(\zeta)}d\zeta})$ is annihilated by 
\begin{align*}
    \hat{P}^\vee_\hbar(\hat{x}^\vee,\hat{y}^\vee,\hat{x}_0^\vee,\hat{y}_0^\vee)=p(\hat{x}_0^\vee)-q(\hat{x}_0^\vee)\hat{y}_0^\vee.
\end{align*}
From the dual version of Corollary \ref{cor:dualquantumcurve}, we find that a quantum curve $\hat{P}_{\hbar}(\hat{x},\hat{y};\hat{x}_0,\hat{y}_0)$ for a generic base point $z_0$ is given by 
\begin{align} \label{eq:QuantumCurve-pq}
    \hat{P}_{\hbar}(\hat{x},\hat{y};\hat{x}_0,\hat{y}_0)=
     p\Big(\hat{y} - \frac{\hbar}{\hat{x} - \hat{x}_0}\Big)-
    q\Big(\hat{y} - \frac{\hbar}{\hat{x} - \hat{x}_0}\Big)
    \Big(
    \hat{x} - \frac{\hbar}{\hat{y} - \hat{y}_0}\Big)
\end{align}
By Lemma \ref{lem:singularpoint}, this reduces further at a singular base point $z_0$ with $x(z_0)=\infty$ to
\begin{align*}
    p\big(\hat{y}\big)-q\big(\hat{y}\big)\bigg(\hat{x}-\frac{\hbar}{\hat{y}-y_0}\bigg).
\end{align*}

\subsubsection{Rational function example via symplectic duality} 
We derive a quantization of exactly the same spectral curve~\eqref{eq:Example-Rational-x} via symplectic  duality.
Let us start with the trivial curve $(\Sigma=\mathbb{P}^1,x^\dagger=z,y^\dagger=z,B= \frac{dz_1\,dz_2}{(z_1-z_2)^2})$ with $P(x^\dagger,y^\dagger)=y^\dagger-x^\dagger$ and the following quantum curves:
\begin{align*}
\hat P^\dagger_{1,\hbar}(\hat x^\dagger,\hat y^\dagger; \hat x^\dagger_0, \hat y^\dagger_0)&=-\hat x^\dagger+\hat y^\dagger,
\\
\hat P^\dagger_{2,\hbar}(\hat x^\dagger,\hat y^\dagger; \hat x^\dagger_0, \hat y^\dagger_0)&=\hat x^\dagger-\hat y^\dagger-\hat x^\dagger_0+\hat y^\dagger_0,
\end{align*}
and obtain the curve~\eqref{eq:Example-Rational-x} via the symplectic transformation 
$$
(x,y)=(x^\dagger+R(y^\dagger),y^\dagger),
$$
where $R(y^\dagger)=\frac{p(y^\dagger)}{q(y^\dagger)}-y^\dagger$.

Applying now Proposition \ref{prop:sympquacurve} with the roles of $(x,y)$ and $(x^\dagger,y^\dagger)$ reversed, we have:
\begin{align*}
& \hat{P}_{1,\hbar} (\hat{x},\hat{y};\hat{x}_0,\hat{y}_0)
\\ 
& =
\Big(\hat{y} - \tfrac{\hbar}{\hat{x} - \hat{x}_{0}}
+  \tfrac{\hbar}{\hat{x} -R(\hat{y} - \tfrac{\hbar}{\hat{x} - \hat{x}_{0}}) -\hat{x}_{0} +R(\hat{y}_{0} - \tfrac{\hbar}{\hat{x}- \hat{x}_{0}}) }\Big)
\\
& \quad
-    \Big(	\hat{x} -R(	\hat{y} - \tfrac{\hbar}{\hat{x} - \hat{x}_{0}})\Big),
\\ 
& =
    \hat{y}-\tfrac{\hbar}{\hat{x}-\hat{x}_0}
    +\tfrac{\hbar}{\hat{x}-\tfrac{p(\hat{y}-\tfrac{\hbar}{\hat{x}-\hat{x}_0})}{q(\hat{y}-\tfrac{\hbar}{\hat{x}-\hat{x}_0})}
    	+\hat{y}-\tfrac{\hbar}{\hat{x}-\hat{x}_0}-\hat{x}_0+\tfrac{p(\hat{y}_0-\tfrac{\hbar}{\hat{x}-\hat{x}_0})}{q(\hat{y}_0-\tfrac{\hbar}{\hat{x}-\hat{x}_0})}-\hat{y}_0+\tfrac{\hbar}{\hat{x}-\hat{x}_0}}\\
    &\quad -\hat{x}+\tfrac{p(\hat{y}-\tfrac{\hbar}{\hat{x}-\hat{x}_0})}{q(\hat{y}-\tfrac{\hbar}{\hat{x}-\hat{x}_0})}-\hat{y}+\tfrac{\hbar}{\hat{x}-\hat{x}_0}\\
    & =\tfrac{p(\hat{y}-\tfrac{\hbar}{\hat{x}-\hat{x}_0})}{q(\hat{y}-\tfrac{\hbar}{\hat{x}-\hat{x}_0})}-\hat{x}+\tfrac{\hbar}{\hat{x}-\tfrac{p(\hat{y}-\tfrac{\hbar}{\hat{x}-\hat{x}_0})}{q(\hat{y}-\tfrac{\hbar}{\hat{x}-\hat{x}_0})}+\hat{y}-\hat{x}_0+\tfrac{p(\hat{y}_0-\tfrac{\hbar}{\hat{x}-\hat{x}_0})}{q(\hat{y}_0-\tfrac{\hbar}{\hat{x}-\hat{x}_0})}-\hat{y}_0}.
\end{align*}
A similar computation implies
\begin{align}\nonumber
	\hat{P}_{2,\hbar} (\hat{x},\hat{y};\hat{x}_0,\hat{y}_0)
=
\hat{x}-\tfrac{p(\hat{y}-\tfrac{\hbar}{\hat{x}-\hat{x}_0})}{q(\hat{y}-\tfrac{\hbar}{\hat{x}-\hat{x}_0})}-\hat{x}_0+\tfrac{p(\hat{y}_0-\tfrac{\hbar}{\hat{x}-\hat{x}_0})}{q(\hat{y}_0-\tfrac{\hbar}{\hat{x}-\hat{x}_0})}.
\end{align}
Both operators vanish the wave function $\psi_{x_0}(x)$. Note that the second operator appears also as a summand in the denominator of the last term of the first operator. Thus, this summand can be replaced by $0$, and we see that the following operator still vanishes the wave functions:
\begin{align*}
\frac{p(\hat{y}-\tfrac{\hbar}{\hat{x}-\hat{x}_0})}{q(\hat{y}-\tfrac{\hbar}{\hat{x}-\hat{x}_0})}
-\hat{x}
+\frac{\hbar}{\hat{y}-\hat{y}_0}.
\end{align*}
After multiplying with $q(\hat{y}_0-\frac{\hbar}{\hat{x}-\hat{x}_0})$ from the left, we recover on the $(x,y)$-side the quantum curve~\eqref{eq:QuantumCurve-pq}.

\subsubsection{Airy curve}
The most fundamental but illuminating example of TR is the Airy curve, see Example \ref{ex:KontWitt}: 
$$
(\Sigma=\mathbb{P}^1,x=z^2,y=z,B=\frac{dz_1\,dz_2}{(z_1-z_2)^2})
$$
The quantization from TR perspective was obtained in \cite{Bergere:2009zm} for the singular base point. Applying the previous result for general $p(y)$ and $q(y)$ with $p(y)=y^2$ and $q(y)=1$, we find
\begin{align*}
    \hat{P}_\hbar(\hat{x},\hat{y};\hat{x}_0,\hat{y}_0)=\bigg(\hat{y}-\frac{\hbar}{\hat{x}-\hat{x}_0}\bigg)^2-\bigg(\hat{x}-\frac{\hbar}{\hat{y}-\hat{y}_0}\bigg),
\end{align*}
which indeed annihilates the wave function 
\begin{align*}
    \psi_{x_0}(x)=Ai_+(x)Ai'_-(x_0)-Ai_+'(x)Ai_-(x_0)
\end{align*}
with the asymptotics of the Airy function given by
\begin{align*}
Ai_{\pm}(x)=\frac{e^{\mp\frac{2}{3\hbar} x^{3/2}}}{\sqrt{2\pi} x^{1/4}}\sum_{k=0}^\infty \frac{(6k)!}{1296^k(2k)!(3k)!}\left(\mp\frac{\hbar}{x^{3/2}}\right)^k.
\end{align*}

\subsubsection{Bessel curve}
The Bessel curve 
$$
(\Sigma=\mathbb{P}^1,x=z^2,y=\frac{1}{z},B=\frac{dz_1\,dz_2}{(z_1-z_2)^2})
$$
is of CN-TR type since $y$ is singular at the ramification point of $x$. Applying the previous result for general $p(y)$ and $q(y)$ with $p(y)=1$ and $q(y)=y^2$, we find
\begin{align*}
    \hat{P}_\hbar(\hat{x},\hat{y};\hat{x}_0,\hat{y}_0)=1-\bigg(\hat{y}-\frac{\hbar}{\hat{x}-\hat{x}_0}\bigg)^2\bigg(\hat{x}-\frac{\hbar}{\hat{y}-\hat{y}_0}\bigg).
\end{align*}
This operator annihilates the general wave function $\psi_{x_0}(x)$ of the Bessel curve.

Taking the limit $z_0\to \infty$ which gives $x_0\to \infty$ for the base point corresponds to the limit $y_0\to 0$. Replacing this via Lemma \ref{lem:singularpoint} in the quantum curve operator, we get 
\begin{align*}
    1-\hat{y}^2\bigg(\hat{x}-\frac{\hbar}{\hat{y}}\bigg)=1-\hat{y}\hat{x}\hat{y},
\end{align*}
which is the quantum curve of~\cite{Do:2016odu}.

\subsection{BE-TR-like examples}\label{Sec:Ex-rspin}
The previous results for coprime polynomials $p(y)$ and $q(y)$ can also be applied for curves with higher order ramification points. 

For the so-called $r$-spin spectral curve $(\Sigma=\mathbb{P}^1,x=z^r,y=z,B=\frac{dz_1\,dz_2}{(z_1-z_2)^2})$ with $p(y)=y^r$ and $q(y)=1$ we find
\begin{align*}
    \hat{P}(\hat{x},\hat{y};\hat{x}_0,\hat{y}_0)=\bigg(\hat{y}-\frac{\hbar}{\hat{x}-\hat{x}_0}\bigg)^r-\bigg(\hat{x}-\frac{\hbar}{\hat{y}-\hat{y}_0}\bigg)
\end{align*}
for the general base point $x(z_0)=x_0$. 

Taking the limit $z_0\to \infty$ which gives $x_0\to \infty$ for the base point corresponds to the limit $y_0\to \infty$. Replacing this in the quantum curve operator, we get 
\begin{align*}
    \hat{P}_\hbar(\hat{x},\hat{y})=\hat{y}^r-\hat{x},
\end{align*}
which is the quantum curve of~\cite{JEP_2017__4__845_0}.

For the so-called negative $r$-spin spectral curve $(\mathbb{P}^1,x=z^r,y=\frac{1}{z},\frac{dz_1\,dz_2}{(z_1-z_2)^2})$ with $p(y)=1$ and $q(y)=y^r$, we find
\begin{align*}
    \hat{P}(\hat{x},\hat{y};\hat{x}_0,\hat{y}_0)=1-\bigg(\hat{y}-\frac{\hbar}{\hat{x}-\hat{x}_0}\bigg)^r\bigg(\hat{x}-\frac{\hbar}{\hat{y}-\hat{y}_0}\bigg)
\end{align*}
for general base point $x(z_0)=x_0$. 
Taking the limit $z_0\to \infty$ which gives $x_0\to \infty$ for the base point corresponds to the limit $y_0\to 0$. Replacing this via Lemma \ref{lem:singularpoint} in the quantum curve operator, we get 
\begin{align*}
    \hat{P}_\hbar(\hat{x},\hat{y})=1-\hat{y}^r\bigg(\hat{x}-\frac{\hbar}{\hat{y}}\bigg)=1-\hat{y}^{r-1}\hat{x}\hat{y},
\end{align*}
which is the quantum curve of~\cite{chidambaram2023relationsoverlinemathcalmgnnegativerspin}.

\subsection{Log-TR-like examples}
In this subsection, we derive examples of quantum curves related to Log-TR. It is important to notice that a bunch of examples with logarithmic singularities 
can be derived from CEO-TR. These examples are still valid since all logarithmic singularities in those examples are not Log-TR-vital. However, if we want to derive quantum curves from the $x$-$y$ dual curve or the symplectic dual curve, the dual side has Log-TR-vital points, and we need to take the Log-TR contribution into account. 

If Log-TR is used, either $x$ or $y$ is locally described by a logarithm. This implies that we have exponential terms in the equation of the spectral curve. Note that quantization of an exponential term generates a formal shift operator, for instance, for $e^y\to e^{\hat{y}}=e^{\hbar\frac{d}{dx}}$ we have
\begin{align}\label{exphderivative}
    e^{\hat{y}}f(x)=f(x+\hbar).
\end{align}

All examples below have considerable meaning in different areas of mathematics or mathematical physics and for many of them their quantum curves were studied in the literature. It worth to mention that historically there were two approaches used to derive quantum curves based on their relation to KP integability:
\begin{itemize}
    \item In some cases, one can use the idea proposed in~\cite{zhou2012quantummirrorcurvesmathbb}. In a nutshell, sometimes it is possible to explicitly compute the wave function and then compare the terms of its $\hbar$-expansion using suitable differential operators. 
    \item Another idea is discussed in~\cite{ALS-Quantum,ACEH-main}: one can employ the Kac-Schwarz operators characterizing the point in the semi-infinite Grassmannian in order to derive an equation for the leading basis vector.
\end{itemize}
Both approaches require a manual tuning to produce meaningful formulas.

\subsubsection{$r$-spin $q$-orbifold Hurwitz numbers}
We give an example of considerable interest in enumerative geometry, especially within the theory of Hurwitz numbers, which is a counting problem of special types of ramified coverings. For a detailed enumerative interpretation, we refer to \cite{Mulase:2013pa,DKPS-qr-ELSV}.

Take the spectral curve of the so-called $r$-spin $q$-double Hurwitz numbers $(\mathbb{P}^1, x, y, \frac{dz_1\,dz_2}{(z_1-z_2)^2})$ with $x, y$ parametrized by
\begin{align}\label{hurwitzrq}
    x = \log z - z^{rq}, \qquad y = z^q.
\end{align}
This spectral curve falls into the category of Log-TR. We aim to quantize this curve by applying the symplectic duality transformation of a quantum curve as described in Proposition~\ref{prop:sympquacurve}. Let us start with the following simpler curve:
\begin{align*}
    x^\dagger = \log z, \qquad y^\dagger = z^q,
\end{align*}
which is related to the curve~\eqref{hurwitzrq} via $x = x^\dagger - (y^\dagger)^r$ and $y = y^\dagger$. The curve $(x^\dagger, y^\dagger)$ has vanishing differentials $\omega_{g,n}^\dagger$ except for $(g,n) = (0,1), (0,2)$. The wave function is given by
\begin{align*}
    \psi^\dagger_{x^\dagger(z_0)}(x^\dagger(z)) = \exp\left(\frac{1}{\hbar q}(z^q - z_0^q) + \frac{1}{2}\log(z z_0) + \log\frac{\log z - \log z_0}{z - z_0}\right),
\end{align*}
and after regularization at $z_0 = \infty$:
\begin{align}
    \psi^{\dagger,\mathrm{reg}}_{\infty}(x^\dagger(z)) = \sqrt{z} \, e^{\frac{1}{\hbar q}z^q},
\end{align}
which is annihilated by the quantum curve
\begin{align}
    \hat{y}^\dagger - \frac{\hbar}{2} - e^{q\hat{x}^\dagger}.
\end{align}
Applying Proposition~\ref{prop:sympquacurve} together with Lemma~\ref{lem:singularpoint}, the resulting quantum curve for the $r$-spin $q$-double Hurwitz numbers reads
\begin{align}\label{quantumcurverspinqdoubleHurwitz}
    \hat{y} - \frac{\hbar}{2} - e^{q(\hat{x} - \hat{y}^r)}.
\end{align}

\subsubsection{Comparison with the previous version}
To compare the quantum curve~\eqref{quantumcurverspinqdoubleHurwitz}
with the one previously derived in \cite{Mulase:2013pa}, we have to multiply the wave function by $e^{\frac{x}{2}}$, which is equivalent to multiplying the quantum curve by $e^{-\frac{\hat{x}}{2}}$ from the right. A further multiplication from the left by $e^{\frac{\hat{x}}{2}}$ yields:
\begin{align*}
    e^{\frac{\hat{x}}{2}}\left(\hat{y} - \frac{\hbar}{2} - e^{q(\hat{x} - \hat{y}^r)}\right)e^{-\frac{\hat{x}}{2}}
    &= \hat{y} - e^{\frac{\hat{x}}{2}} e^{q(\hat{x} - \hat{y}^r)} e^{-\frac{\hat{x}}{2}} \\
    &= \hat{y} - e^{\frac{\hat{x}}{2} + q\hat{x}} \, e^{\frac{1}{\hbar(r+1)}\left((\hat{y} - q\hbar)^{r+1} - \hat{y}^{r+1}\right)} e^{-\frac{\hat{x}}{2}}.
\end{align*}
This coincides exactly with the quantum curve derived in \cite[Eq.~(58)]{Mulase:2013pa}. The final computational step follows from the identity
\begin{align}\label{identityrspinhurwitz}
    e^{q\hat{x}} \, e^{\frac{1}{\hbar(r+1)}\left((\hat{y} - q\hbar)^{r+1} - \hat{y}^{r+1}\right)} = e^{q(\hat{x} - \hat{y}^r)},
\end{align}
which is a special case of the Baker–Campbell–Hausdorff formula. The identity \eqref{identityrspinhurwitz} can be proven by showing that both sides satisfy the same first-order differential equation in the variable \( q \), with the same initial condition at \( q = 0 \), where both operators restrict to identity. The differential equation for the right-hand side reads
\begin{align*}
    \frac{d}{dq} e^{q(\hat{x} - \hat{y}^r)} = (\hat{x} - \hat{y}^r) \, e^{q(\hat{x} - \hat{y}^r)},
\end{align*}
while for the left-hand side we compute
\begin{align*}
    &\frac{d}{dq} \left[ e^{q\hat{x}} \, e^{\frac{1}{\hbar(r+1)}\left((\hat{y} - q\hbar)^{r+1} - \hat{y}^{r+1}\right)} \right]\\
    &= \hat{x} \, e^{q\hat{x}} \, e^{\frac{1}{\hbar(r+1)}\left((\hat{y} - q\hbar)^{r+1} - \hat{y}^{r+1}\right)} \\
    &\quad + e^{q\hat{x}} \cdot \left( \tfrac{-\hbar}{\hbar(r+1)} \cdot (r+1)(\hat{y} - q\hbar)^r \right) \cdot e^{\frac{1}{\hbar(r+1)}\left((\hat{y} - q\hbar)^{r+1} - \hat{y}^{r+1}\right)} \\
    &= (\hat{x} - \hat{y}^r) \, e^{q\hat{x}} \, e^{\frac{1}{\hbar(r+1)}\left((\hat{y} - q\hbar)^{r+1} - \hat{y}^{r+1}\right)}.
\end{align*}
Here we have used that commuting \( e^{q\hat{x}} \) through a function \( f(\hat{y}) \) results in a shift of the form \( f(\hat{y} + q\hbar) \). Thus, both sides of \eqref{identityrspinhurwitz} satisfy the same differential equation with the same initial condition, proving the identity \eqref{identityrspinhurwitz}. 

Our derivation of the quantum curve via symplectic duality provides new insight into the structure of the quantum curve for the \( r \)-spin \( q \)-double Hurwitz numbers. The intricate representation given in \cite[Eq.~(58)]{Mulase:2013pa} is shown to be equivalent to the canonical form \eqref{identityrspinhurwitz}, which is not immediately apparent from the original derivation in \textit{loc.~cit.}, but becomes clear via symplectic duality.

\subsubsection{Colored HOMFLY-PT polynomials of  torus knots}
Here we recall an example related to knot theory. Due to important achievements in mathematical physics, it is well known that knot theory has a deep connection to topological string theory and Chern–Simons theory. We will not explore these connections in detail here, but simply cite the spectral curve associated with the $(P,Q)$-torus knot \cite{Brini:2011wi}.
The original derivation of the spectral curve in \cite{Brini:2011wi} is based on an \( SL(2,\mathbb{Z}) \) action on \( (x,y) \) originating from the resolved conifold. This \( SL(2,\mathbb{Z}) \) action is essentially a combination of \( x \)-\( y \) duality and symplectic duality transformations. The combinatorial structure of colored HOMFLY-PT polynomials for the $(P,Q)$-torus knot was analyzed in \cite{Dunin-Barkowski:2017zsd}, where the quantum curve was derived. The corresponding partition function, the so-called extended Ooguri–Vafa, is the expansion of a system of differentials that satisfy TR~\cite{Dunin-Barkowski:2020sez}, see also~\cite{Bychkov:2020yzy}.

A convenient representation of the spectral curve \( (\mathbb{P}^1, x, y, \frac{dz_1\,dz_2}{(z_1 - z_2)^2}) \) is given by 
\begin{align}\label{PQknotcurve}
	x &= \log z - \frac{P}{Q} y, \\\nonumber 
	y &= \log(1 - A^{-1} z) - \log(1 - A z)
\end{align}
(see for instance \cite{Alexandrov:2016edh,Dunin-Barkowski:2020sez}).
The relation of this curve to that in \cite{Brini:2011wi} was discussed in detail in \cite{Dunin-Barkowski:2020sez}.

The simpler corresponding symplectic dual spectral curve reads:
\begin{align*}
	x^\dagger &= \log z, \\\nonumber 
	y^\dagger &= \log(1 - A^{-1} z) - \log(1 - A z).
\end{align*}
This curve is quantized via \eqref{quantumcurvedaggerLogTR} as
\begin{align*}
	e^{\frac{\hat{y}^\dagger}{2}} - e^{\frac{\hbar}{2}} \, \frac{1 - A^{-1} e^{\hat{x}^\dagger}}{1 - A e^{\hat{x}^\dagger}} \, e^{-\frac{\hat{y}^\dagger}{2}} .
\end{align*}
Performing the symplectic duality transformation
\begin{align*}
	x = x ^\dagger- \frac{P}{Q} y^\dagger,
\end{align*}
together with Proposition~\ref{prop:sympquacurve} and Lemma~\ref{lem:singularpoint}, we obtain:
\begin{align}\label{quantumcurvePQ}
	e^{\frac{\hat{y}}{2}} - e^{\frac{\hbar}{2}} \, \frac{1 - A^{-1} e^{\hat{x} - \frac{P}{Q} \hat{y}}}{1 - A e^{\hat{x} - \frac{P}{Q} \hat{y}}} \, e^{-\frac{\hat{y}}{2}}.
\end{align}

\subsubsection{Comparison with the previous version}
To match \eqref{quantumcurvePQ} with the result derived via the Fock space representation in~\cite{Dunin-Barkowski:2017zsd}, we again sandwich the quantum curve with \( e^{\hat{x}/2} \) from the left and \( e^{-\hat{x}/2} \) from the right, as in the \( r \)-spin \( q \)-orbifold example. This, followed by multiplication with \( 1 - A e^{\hat{x} + \frac{P}{Q} \hat{y}} \), leads to:
\begin{align*}
	& e^{\hat{x}/2}(1 - A e^{\hat{x} - \frac{P}{Q} \hat{y}})  \left( e^{\frac{\hat{y}}{2}} - e^{\frac{\hbar}{2}} \, \frac{1 - A^{-1} e^{\hat{x} - \frac{P}{Q} \hat{y}}}{1 - A e^{\hat{x} - \frac{P}{Q} \hat{y}}} \, e^{-\frac{\hat{y}}{2}} \right) e^{-\hat{x}/2} \\
	& ={} (1 - A e^{\hat{x} - \frac{P}{Q} \hat{y} - \frac{P}{Q} \frac{\hbar}{2}}) e^{\frac{\hat{y}}{2} + \frac{\hbar}{4}} - e^{\frac{\hbar}{2}} (1 - A^{-1} e^{\hat{x} - \frac{P}{Q} \hat{y} - \frac{P}{Q} \frac{\hbar}{2}}) e^{-\frac{\hat{y}}{2} - \frac{\hbar}{4}} \\
	& ={} e^{\hbar/4} \left[ \left( e^{\frac{\hat{y}}{2}} - e^{\frac{\hat{y}}{2}} \right) - e^{\hat{x} - \frac{P}{Q} \hat{y} - \frac{P}{Q} \frac{\hbar}{2}} \left( A e^{\frac{\hat{y}}{2}} - A^{-1} e^{\frac{\hat{y}}{2}} \right) \right].
\end{align*}

Note that \( e^{\hat{x} - \frac{P}{Q} \hat{y} - \frac{P}{Q} \frac{\hbar}{2}} = e^{\hat{x}} e^{-\frac{P}{Q} \hat{y}} \) by the Baker–Campbell–Hausdorff formula. Thus, we recover \cite[Thm.~10.1]{Dunin-Barkowski:2017zsd} after the trivial shift \( x \mapsto x + \frac{P}{Q} \log A \).

\begin{remark}
It is remarkable how the initial factor of \( e^{\hbar/2} \) in \eqref{quantumcurvePQ} is canceled through the application of the Baker–Campbell–Hausdorff formula and the commutation of \( e^{\hat{x}/2} \). However, it is precisely this computation that previously prevented a conceptual understanding of such examples within a more general framework. This is now achieved in the context of symplectic duality for quantum curves.
\end{remark}

\subsubsection{General Log-TR example}

The example of $r$-spin $q$-orbifold Hurwitz numbers and colored HOMFLY-PT invariants of the torus knots fall into a very general 
family of examples arising from Orlov–Scherbin partition functions, also known as hypergeometric KP tau functions. There is a number of general approaches that proves for all reasonable cases that the corresponding $n$-point functions can be assembled into a system of differentials that satisfy topological recursion on $(\Sigma=\mathbb{P}^1,x(z),y(z),B=\frac{dz_1dz_2}{(z_1-z_2)^2})$ with
\begin{align} \label{eq:GenHypergeom}
x(z) = \log z - \psi(r(z)), y(z)=r(z),
\end{align}
where $r(z)$ and $\psi(z)$ are some functions on  $\mathbb{P}^1$  such that both $dx(z)$ and $dy(z)$ are meromorphic differentials. This implies some restrictions on $r$ and $\psi$, and there are several very general families of examples, so-called Families I, II, and III, considered in~\cite{Alexandrov:2023tgl-LogTR}. 

This spectral curve data is obtained from a simpler spectral curve $(x^\dagger,y^\dagger)=(\log z, r(z))$ via a symplectic transformation $(x^\dagger,y^\dagger)\mapsto (x,y)=(x^\dagger - \psi(y^\dagger), y^\dagger)$, and we can use it in combination with Proposition~\ref{prop:sympquacurve} to derive the quantum curve for~\eqref{eq:GenHypergeom}. Here the specifics of the functions $r$ and $\psi$ start to play a role.

Above we described how it works for $r$-spin $q$-orbifold Hurwitz numbers and colored HOMFLY-PT invariants of the torus knots, which are illustrative for Family I and II, respectively. Now we illustrate this derivation of a quantum curve in application to the an extended version of Family III in \cite{Alexandrov:2023tgl-LogTR} and corresponds to Case III in \cite{Alexandrov:2024ajj}, described by the spectral curve 
\begin{align} \label{eq:FamIII}
    x(z) &= \log z - \left(T(e^{\gamma_1 y}, \ldots, e^{\gamma_K y}) + \kappa y\right), \\ \notag 
    y(z) &= \sum_i \frac{1}{\alpha_i} \log(z - a_i),
\end{align}
where \( T \) is some rational function, and \( a_i, \alpha_i, \kappa \) are constants such that \( \frac{\gamma_i}{\alpha_j} \in \mathbb{Z} \), and \( dx \) is meromorphic.

First, we consider the simpler spectral curve \( (x^\dagger, y^\dagger) \) of the form
\begin{align*}
    x^\dagger(z) = \log z, \qquad 
    y^\dagger(z) = \sum_i \frac{1}{\alpha_i} \log(z - a_i).
\end{align*}
For this curve, all \( \omega_{g,n}^\dagger \) are explicitly given by Log-TR and take the form
\begin{align*}
    \omega_{g,n}^\dagger = [\hbar^{2g}]\, dx \, \frac{\delta_{n,1}}{\cS(\hbar \partial_x)} y + \frac{dz_1\, dz_2}{(z_1 - z_2)^2}.
\end{align*}
Looking at the wave function, the Bergman kernel \( \omega_{0,2} \) contributes a factor of \( \sqrt{z} = e^{\frac{x}{2}} \). Taking the logarithm of the wave function (with the base point \( z_0 = \infty \) and after a regularization), we find:
\begin{align*}
    && \log \psi^\dagger(x^\dagger) &= \frac{x^\dagger}{2} + \sum_{g=0}^\infty \hbar^{2g-1} [\hbar^{2g}] \int^{x^\dagger} dx^\dagger \, \frac{1}{\cS(\hbar \partial_{x^\dagger})} y^\dagger .
\end{align*}
Through the following chain of equivalent relations
\begin{align*}
    && \left(e^{\frac{\hbar}{2} \partial_{x^\dagger}} - e^{-\frac{\hbar}{2} \partial_{x^\dagger}} \right) \log \psi^\dagger(x^\dagger)
    &= \left(e^{\frac{\hbar}{2} \partial_{x^\dagger}} - e^{-\frac{\hbar}{2} \partial_{x^\dagger}} \right) \frac{x^\dagger}{2} + y^\dagger(e^{x^\dagger}) ,\\
     && \log \psi^\dagger(x^\dagger + \tfrac{\hbar}{2}) - \log \psi^\dagger(x^\dagger - \tfrac{\hbar}{2}) 
    &= \hbar/2 + y^\dagger(e^{x^\dagger}) ,\\
 && \psi(x^\dagger + \tfrac{\hbar}{2}) &= e^{\frac{\hbar}{2} + y^\dagger(e^{x^\dagger})} \, \psi(x^\dagger - \tfrac{\hbar}{2}),
\end{align*}
we conclude that the quantum curve takes the form
\begin{align}\label{quantumcurvedaggerLogTR}
    e^{\frac{\hat{y}^\dagger}{2}} - e^{\frac{\hbar}{2} + y^\dagger(e^{\hat{x}^\dagger})} \, e^{-\frac{\hat{y}^\dagger}{2}} .
\end{align}
Note that \( y^\dagger(e^{\hat{x}^\dagger}) \) in the exponential indicates that the function \( y^\dagger \) is evaluated at \( e^{\hat{x}^\dagger} \), no operator \( \hat{y}^\dagger \) appears here.

Performing the next step, the symplectic duality transformation
\begin{align*}
    x = x^\dagger - \left(T(e^{\gamma_1 y^\dagger}, \ldots, e^{\gamma_K y^\dagger}) + \kappa y^\dagger\right), \qquad y = y^\dagger,
\end{align*}
amounts, by Proposition~\ref{prop:sympquacurve} together with Lemma~\ref{lem:singularpoint}, to the quantum curve
\begin{align}\label{quantumcurvegeneralLogTR}
    e^{\frac{\hat{y}}{2}} - e^{\frac{\hbar}{2} + y\left(e^{\hat{x} + T(e^{\gamma_1 \hat{y}}, \ldots, e^{\gamma_K \hat{y}}) + \kappa \hat{y}}\right)} e^{-\frac{\hat{y}}{2}} .
\end{align}

We should mention that if \( \alpha_i \neq \pm1 \), roots of operators appear. It is not clear to us whether such expressions can actually be transformed into operators without roots by multiplying with additional operators on the quantum curve. To our knowledge, no such examples have been studied in the literature. 

To illustrate this, consider the example:
\begin{align*}
    x = \log z - y, \qquad y = \frac{1}{2} \log(z - 1).
\end{align*}
The corresponding quantum curve then takes the form:
\begin{align*}
    e^{\frac{\hat{y}}{2}} - e^{\frac{\hbar}{2} + y\left(e^{\hat{x} + \hat{y}}\right)} e^{-\frac{\hat{y}}{2}}
    = e^{\frac{\hat{y}}{2}} - e^{\frac{\hbar}{2}} \left(e^{\hat{x} + \hat{y}} - 1\right)^{1/2} e^{-\frac{\hat{y}}{2}}.
\end{align*}

\begin{remark}
	Note that the Families I, II, and III are not non-intersecting. In particular,  
	the spectral curve \eqref{PQknotcurve} is also a special case of~\eqref{eq:FamIII} with parameters \( T = 0 \), \( a_1 = A \), \( \alpha_1 = 1 \), \( a_2 = A^{-1} \), \( \alpha_2 = -1 \). Note that \( \sum_i \frac{1}{\alpha_i} = 0 \), which leads to the following behavior  at the base point \( z_0 = \infty \), \( y_0 = y(z_0) = -2\log A \), which is not divergent. However, the base point \( z_0=\infty \) is still singular for \( x \), and applying Proposition~\ref{prop:sympquacurve}, the divergence of \( x \) at this point yields, by Lemma~\ref{lem:singularpoint}, a simple transformation.
\end{remark}

\subsubsection{Gaiotto curve}
Take the curve from \cite{Borot:2024uos} and change their $x$ and $y$ to 
\begin{align*}
    x\mapsto \tilde{x}=\log x,\qquad y\to \tilde{y}=y\cdot x=z.
\end{align*}
This transformation leaves the symplectic form invariant $dx\wedge dy=d\tilde{x}\wedge d\tilde{y}$. 
More importantly, the 1-form $\omega_{0,1}$ is invariant under this transformation
\begin{align*}
    \omega_{0,1}=y\,dx=\tilde{y}\,d\tilde{x}.
\end{align*}
All algebraic ramification points of $x$ and $\tilde{x}$ do also coincide. Thus, all $\omega_{g,n}$ for $(x,y)$ coincide with $\omega_{g,n}$ of $(\tilde{x},\tilde{y})$. The benefit of working with $\tilde{x},\tilde{y}$ is that the $x$-$y$ dual side of the spectral curve, that is $(\mathbb{P}^1,x^\vee=\tilde{y},y^\vee=\tilde{x},\frac{dz_1\,dz_2}{(z_1-z_2)^2})$ produces with Log-TR explicit expressions for $\omega_{g,n}^\vee$. 

Being more explicit, the new curve reads
\begin{align}\label{Gaiottocurve}
    \tilde{x}(z) = y^\vee(z) =&\log\bigg(-\Lambda^r \frac{\prod_{a=1}^r (P_a+z)}{\prod_{b=1}^{r-1} (Q_a-z)}\bigg);\\ \notag 
    \tilde{y}(z) = x^\vee(x) =&z.
\end{align}
The $x$-$y$ dual multi-differentials read
\begin{align*}
    \omega_{g,n}^\vee=[\hbar^{2g}]\delta_{n,1}\bigg(\frac{1}{\, S(\hbar  \partial_{x^\vee(z)})}y^\vee(z) \bigg)dz+\delta_{g,0}\delta_{n,2}\frac{ dz_1\,dz_2}{(z_1-z_2)^2}.
\end{align*}
The $x$-$y$ dual wave function satisfies the difference equation (no contribution from $\omega_{0,2}^\vee$ since $x^\vee=z$)
\begin{align*}
    -\log \psi^\vee_{x^\vee_0+\frac \hbar 2}(x^\vee)+\log \psi^\vee_{x^\vee_0-\frac \hbar 2}(x^\vee)=y^\vee(x_0^\vee)
\end{align*}
and in operator form 
\begin{align*}
    \bigg(e^{\frac{\hat{y}_0^\vee}{2}}+\Lambda^r \frac{\prod_{a=1}^r (P_a+\hat{x}_0^\vee)}{\prod_{b=1}^{r-1} (Q_a-\hat{x}_0^\vee)}e^{\frac{-\hat{y}_0^\vee}{2}}\bigg)\psi^\vee_{x^\vee_0}(x^\vee)=0.
\end{align*}
Multiplying with $\prod_{b=1}^{r-1} (Q_a-\hat{x}_0^\vee)$ from the left and shifting the wave function $x_0^\vee\to x_0^\vee+\frac{\hbar}{2}$, we can find the representation
\begin{align*}
    \underbrace{\bigg(\prod_{b=1}^{r-1} (Q_a-\hat{x}_0^\vee) e^{\hat{y}_0^\vee}+\Lambda^r \prod_{a=1}^r (P_a+\hat{x}_0^\vee)\bigg)}_{\hat{P}^\vee_\hbar(\hat{x}^\vee,\hat{y}^\vee;\hat{x}_0^\vee,\hat{y}_0^\vee)}\psi^\vee_{x^\vee_0+\frac\hbar 2}(x^\vee)=0.
\end{align*}
The next step is to proceed to the $x$-$y$ dual side via Corollary \ref{cor:dualquantumcurve}. We find the quantum curve
\begin{align*}
	\hat{P}_\hbar(\hat{x},\hat{y};\hat{x}_0,\hat{y}_0)= \prod_{b=1}^{r-1} \big(Q_a-\hat{\tilde{y}}+\frac{\hbar}{\hat{\tilde x} - \hat{\tilde x}_0}\big) e^{\hat{\tilde x}-\frac{\hbar}{\hat{\tilde y}-\hat{\tilde y}_0}}+\Lambda^r \prod_{a=1}^r \big(P_a+\hat{\tilde y}-\frac{\hbar}{\hat{\tilde x} - \hat{\tilde x}_0}\big)
\end{align*}
Consider the base point $z_0=Q_c$ which is a pole of $\tilde{x}$ but not of $\tilde{y}$. For the $b=c$ in the first product, we replace $Q_c$ by $\hat{\tilde y}_0$, and further using Lemma \ref{lem:singularpoint} we get:
\begin{align}\label{Gaiottocurve1}
	(\hat{\tilde y}_0-\hat{\tilde y})\prod_{\substack{b=1 \\ b\neq c}}^{r-1} (Q_a-\hat{\tilde y}) e^{\hat{\tilde x}-\frac{\hbar}{\hat{\tilde y}-\hat{\tilde y}_0}}+\Lambda^r \prod_{a=1}^r (P_a+\hat{\tilde y})
\end{align}
Notice that
\begin{align*}
    &e^{\hat{\tilde x}-\frac{\hbar}{\hat{\tilde y}-\hat{\tilde y}_0}}=e^{\frac{\hat{\tilde y}-\hat{\tilde y}_0}{\hat{\tilde y}-\hat{\tilde y}_0}\hat{\tilde x}-\frac{\hbar}{\hat{\tilde y}-\hat{\tilde y}_0}}=e^{\frac{1}{\hat{\tilde y}-\hat{\tilde y}_0}\hat{\tilde x}(\hat{\tilde y}-\hat{\tilde y}_0)}=\tfrac{1}{\hat{\tilde y}-\hat{\tilde y}_0}e^{\hat{\tilde x}}(\hat{\tilde y}-\hat{\tilde y}_0)
    =&\tfrac{1}{\hat{\tilde y}-\hat{\tilde y}_0}(\hat{\tilde y}-\hbar-\hat{\tilde y}_0)e^{\hat{\tilde x}}.
\end{align*}
Inserting this identity in \eqref{Gaiottocurve1}, we finally obtain the following operator
\begin{align*}
 \prod_{b=1}^{r-1} (Q_b+\hbar\delta_{c,b}-\hat{\tilde y}) e^{\hat{\tilde x}}+\Lambda^r \prod_{a=1}^r (P_a+\hat{\tilde y})
\end{align*}
which coincides with the quantum curve derived in \cite[Prop. 5.12]{Borot:2024uos}.

\subsection{Gen-TR-like example}
Gen-TR is a very recent generalization of topological recursion and its application to quantum curves was not considered yet. In the following, we give the simplest example of a quantum curve which can not be obtained from CEO-TR but from Gen-TR. The difference is essentially in the $\hbar$-correction terms of the quantum curve. We also give a primer of computation for the so-called $(r,s)$-curves, where we work out the $s=2$ case.

\subsubsection{New version of the Airy curve}
We want to represent another version of the quantization of the Airy curve $(x=z^2,y=z)$ by using Gen-TR. Then only special point is at $z=0$. Consider a Gen-TR spectral curve given as 
\begin{align*}
	(\Sigma=\mathbb{P}^1,dx=2z\, dz,dy=dz,B=\frac{dz_1\,dz_2}{(z_1-z_2)^2},\mathcal{P}=\emptyset)
\end{align*}
and the $x$-$y$ dual spectral curve
 $$
 (\Sigma=\mathbb{P}^1,dx^\vee=dz,dy^\vee=2z\,dz,B=\frac{dz_1\,dz_2}{(z_1-z_2)^2},\mathcal{P}^\vee=\{0\}).
 $$

Note that the set $\mathcal{P}$ is chosen empty rather than to be the set of the critical points of, which would recover CEO-TR. Since $\mathcal{P}=\emptyset$, we find that all stable $\omega_{g,n}$ obtained from Gen-TR vanish, that is we have the following explicit formulas:
\begin{align*}
    \omega_{0,1}=y\,dx=2z^2\,dz,\quad \omega_{0,2}=\frac{dz_1\,dz_2}{(z_1-z_2)^2}.
\end{align*}
The wave function is simply obtained 
\begin{align*}
    \psi_{x_0}(x)=\exp\bigg(\frac{2}{\hbar3}\big(x^{3/2}-x_0^{3/2}\big)-\frac{1}{2}\log(2\sqrt{x}\cdot 2 \sqrt{x_0})+\log(\sqrt{x}+\sqrt{x_0})\bigg).
\end{align*}
This wave function satisfies the differential equation
\begin{align} 
	\bigg(\hat{y_0}^2-\hat{x_0}+\frac{\hbar}{\sqrt{\hat{x}}+\sqrt{\hat{x}_0}}+\hbar^2\frac{3\sqrt{\hat{x_0}}-5\sqrt{\hat{x}}}{16\hat{x}_0^2(\sqrt{\hat{x}}+\sqrt{\hat{x}_0})}\bigg)\psi_{x_0}(x)=0.
\end{align}
The $x$-$y$ dual wave function is annihilated by
\begin{align*}
	&\big(\hat{x}^\vee-\tfrac{\hbar}{\hat{y}^\vee - \hat{y}_0^\vee }\big)^2
	-\hat{y}^\vee + \tfrac{\hbar}{\hat{x}^\vee -\hat{x}_0^\vee}
	+\frac{\hbar}{\sqrt{\hat{y}_0^\vee - \tfrac{\hbar}{\hat{x}^\vee -\hat{x}_0^\vee}}+\sqrt{\hat{y}^\vee - \tfrac{\hbar}{\hat{x}^\vee -\hat{x}_0^\vee}}}
	\\
	&
	+\hbar^2\frac{3\sqrt{\hat{y}^\vee - \tfrac{\hbar}{\hat{x}^\vee -\hat{x}_0^\vee}}-5\sqrt{\hat{y}_0^\vee - \tfrac{\hbar}{\hat{x}^\vee -\hat{x}_0^\vee}}}{16\big(\hat{y}^\vee - \tfrac{\hbar}{\hat{x}^\vee -\hat{x}_0^\vee}\big)^2(\sqrt{\hat{y}_0^\vee - \tfrac{\hbar}{\hat{x}^\vee -\hat{x}_0^\vee}}+\sqrt{\hat{y}^\vee - \tfrac{\hbar}{\hat{x}^\vee -\hat{x}_0^\vee}})}.
\end{align*}
At the singular base point $z_0=\infty$ we have $x_0^\vee=y_0^\vee=\infty$, thus the dual wave function is annihilated at the singular base point (after a suitable regularization and multiplication by $(\hat{y}^\vee)^2$ from the left) by 
\begin{align*}
   (\hat{y}^\vee)^2(\hat{x}^\vee)^2-(\hat{y}^\vee)^3-\hbar^2\frac{5}{16}.
\end{align*}
Note that on the dual side all $\omega_{g,n}^\vee$ are actually non-trivial, see also some explicitly computed coefficients in~\cite[Section 7.3]{Alexandrov:2024tjo}.

\subsubsection{A primer on quantization of the $(r,s)$ curves}
\label{sec:rs-curve}
An important spectral curve is the so-called $(r,s)$ curve, see again~\cite[Section 7.3]{Alexandrov:2024tjo}.  It is especially interesting due to its connections to geometry and $W$-constraints~\cite{CGS}.

For  $r,s\in \mathbb{Z}$ such that $r\geq 2$ and $0<s<r$, we define $x=z^r$ and $y=z^{-s}$. The vanishing locus of the corresponding plane curve reads
\begin{align*}
    P(x,y)=x^{s}y^{r}-1.
\end{align*}

In this case the only Gen-TR  special point is $z=0$. Let us choose the Gen-TR spectral curve as
$$
(\Sigma=\mathbb{P}^1,dx(x)=rz^{r-1}dz,dy(x)=-sz^{-s-1}dz,B=\frac{dz_1\,dz_2}{(z_1-z_2)^2},\mathcal{P}=\{0\})
$$ 
and the $x$-$y$ dual curve as  
$$
(\Sigma =\mathbb{P}^1,dx^\vee(z)=-sz^{-s-1}dz,dy^\vee(z)=rz^{r-1}dz,B=\frac{dz_1\,dz_2}{(z_1-z_2)^2},\mathcal{P}^\vee=\emptyset).
$$ 
Since on the dual side $\mathcal{P}$ is the empty set, we have the explicit formulas
\begin{align*}
    \omega_{0,1}^\vee=-sz^{r-s-1} dz,\quad \omega_{0,2}^\vee=\frac{dz_1\,dz_2}{(z_1-z_2)^2}, \quad
    \omega_{g,n}^\vee=0\quad \text{for all $2g+n-2>0$.}
\end{align*}
The wave function becomes 
\begin{align*}
    \psi^\vee_{x^\vee_0}(x^\vee)=
     \exp\Big(&\frac{s}{\hbar(s-r)}\big((x^\vee)^{\frac{s-r}{s}}-(x_0^\vee)^{\frac{s-r}{s}}\big)-\frac{1}{2}\log(s(x^\vee)^{\frac{s+1}{s}}s(x^\vee_0)^{\frac{s+1}{s}})
     \\
     & +\log\Big(\frac{x^\vee-x^\vee_0}{(x^\vee)^{-\frac 1s}-(x_0^\vee)^{-\frac 1s}}\Big)\Big).
\end{align*}
Trying to find the differential equation solved by this wave function is nontrivial for general $x^\vee,x^\vee_0$. In the limit $x^\vee \to 0$ the wave function gets with appropriate regularization the form
\begin{align*}
    \psi^\vee_{x_0^\vee}(0)=\exp\Big(-\frac{s}{\hbar(s-r)}(x_0^\vee)^{\frac{s-r}{s}}+\frac{s+1}{2s}\log\big((x_0^\vee)\Big).
\end{align*}
We find the operator that vanishes the wave function with a rational exponent for $\hat{x}_0$ to be
\begin{align}\label{pqquantumcurve}
    \hat{y}^\vee_0-(\hat{x}_0^\vee)^{-\frac rs}+\hbar \frac{s+1}{2s\hat{x}^\vee_0}.
\end{align}
For the special case of $s=1$, this operator has integer powers and reduces to the $x$-$y$ dual operator of the quantum curve of the negative $r$-spin example of Sec.~\ref{Sec:Ex-rspin}.

For $s> 1$, we should instead construct an operator with just integer powers by an iterative Galois averaging on the left, replacing~\eqref{pqquantumcurve} by
\begin{align*}
	(\hat{y}^\vee_0-J^{s-1}(\hat{x}^\vee_0)^{-\frac rs}+\hbar \tfrac{s+1}{2s\hat{x}^\vee_0})
	\cdots(\hat{y}^\vee_0-J^{1}(\hat{x}^\vee_0)^{-\frac rs}+\hbar \tfrac{s+1}{2s\hat{x}^\vee_0})
	 (\hat{y}^\vee_0-J^{0}(\hat{x}^\vee_0)^{-\frac rs}+\hbar \tfrac{s+1}{2s\hat{x}^\vee_0}),
\end{align*}
where $J^s=1$ is the primitive $s$-th root of unity, to fix the leading term in $\hbar$, and then analogously proceed to perturbatively fix the higher order terms. 

For instance, for $s=2$, we find
\begin{align*}
	&
    \big(\hat{y}^\vee_0+(\hat{x}_0^\vee)^{-\frac r2}+\tfrac{3\hbar }{4\hat{x}^\vee_0}\big)\big(\hat{y}^\vee_0-(\hat{x}_0^\vee)^{-\frac r2}+ \tfrac{3\hbar}{4\hat{x}^\vee_0}\big)
    \\ &
    =(\hat{y}^\vee_0)^2-(\hat{x}^\vee_0)^{-r} -\tfrac{3\hbar^2}{16 (\hat x_0^\vee)^2}-\tfrac {r\hbar}2 (\hat x_0^\vee)^{-\frac r2-1}.
\end{align*}
The last term can be compensated by adding 
\[
-\tfrac {r\hbar}{2\hat{x}^\vee_0}\big(\hat{y}^\vee_0-(\hat{x}_0^\vee)^{-\frac r2}+ \tfrac{3\hbar}{4\hat{x}^\vee_0}\big),
\]
which makes the final operator vanishing the wave function to be
\begin{align*}
(\hat{y}^\vee_0)^2-(\hat{x}^\vee_0)^{-r} -\tfrac {r\hbar}{2\hat{x}^\vee_0}\hat{y}^\vee_0 -\tfrac{9\hbar^2}{16 (\hat x_0^\vee)^2}.
\end{align*}

The $x$-$y$ duality gives in this example the following quantum curve on the non-trivial side for the singular base point $x_0=\infty$, $y_0=0$ (we also multiplied it by $\hat y^r$ from the left):
\begin{align*}
	\hat y^{r}\big(\hat{x}-\frac{\hbar}{\hat{y}}\big)^2-1-\frac {r\hbar}{2}\hat y^{r-1}\big(\hat x-\frac{\hbar}{\hat y}\big) -\frac{9\hbar^2}{16}\hat y^{r-2}.
\end{align*}

\newcommand{\etalchar}[1]{$^{#1}$}

\end{document}